\documentclass{scrartcl}
\pdfoutput=1

\usepackage{cmap}
\usepackage[utf8]{inputenc}
\usepackage[T1]{fontenc}

\usepackage{amsmath}
\usepackage{amssymb}
\usepackage{amsthm}
\usepackage{thmtools}
\usepackage{bm}
\usepackage{mathrsfs}
\usepackage{csquotes}
\usepackage{tabularx}
\usepackage{enumitem}
\usepackage{hyperref}
\usepackage{bbm}
\usepackage{calc}
\usepackage{tikz}
\usetikzlibrary{automata, calc, matrix, shapes, positioning, decorations.pathreplacing}

\usepackage{pdfpages}

\setcapindent{0pt}

\theoremstyle{plain}
\newtheorem{thm}{Theorem}[section]
\newtheorem{lemma}[thm]{Lemma}
\newtheorem{prop}[thm]{Proposition}
\newtheorem{cor}[thm]{Corollary}

\theoremstyle{definition}

\newtheorem{example}[thm]{Example}
\newtheorem{prob}[thm]{Open Problem}

\theoremstyle{remark}
\newtheorem{remark}[thm]{Remark}

\newcommand{\inverse}[1]{\mkern 1.5mu\overline{\mkern-1.5mu#1\mkern-1.5mu}\mkern 1.5mu}

\DeclareFontFamily{U}{mathb}{\hyphenchar\font45}
\DeclareFontShape{U}{mathb}{m}{n}{
  <5> <6> <7> <8> <9> <10> gen * mathb
  <10.95> mathb10 <12> <14.4> <17.28> <20.74> <24.88> mathb12
}{}
\DeclareSymbolFont{mathb}{U}{mathb}{m}{n}
\DeclareFontSubstitution{U}{mathb}{m}{n}
\DeclareMathSymbol{\drsh}{3}{mathb}{"EB}

\newlength{\edgelength}
\newcommand{\trans}[4]{%
  \begin{tikzpicture}[auto, shorten >=1pt, >=latex, baseline=(l.base), inner sep=0pt, outer xsep=0.3333em]
    \node (l) {\ensuremath{#1}};%
    \setlength{\edgelength}{\widthof{\scriptsize\ensuremath{#2/#3}}+0.5cm}%
    \node[base right=\edgelength of l] (r) {\ensuremath{#4}};%
    \path[->] (l.mid east) edge node[inner sep=0pt] {\scriptsize\ensuremath{#2/#3}} (r.mid west);%
  \end{tikzpicture}%
}

\newcommand{\transa}[3]{%
  \begin{tikzpicture}[auto, shorten >=1pt, >=latex, baseline=(l.base), inner sep=0pt, outer xsep=0.3333em]
    \node (l) {\ensuremath{#1}};%
    \setlength{\edgelength}{\widthof{\scriptsize\ensuremath{#2}}+0.5cm}%
    \node[base right=\edgelength of l] (r) {\ensuremath{#3}};%
    \path[->] (l.mid east) edge node[inner xsep=0pt, inner ysep=0.2em] {\scriptsize\ensuremath{#2}} (r.mid west);%
  \end{tikzpicture}%
}

\newcommand{\wang}[4]{%
  \begin{tikzpicture}[baseline={([yshift={-0.5ex}]current bounding box.east)}, inner sep=0pt, outer sep=0.1ex]%
    \matrix[row sep=0.1ex, column sep=0.1ex, ampersand replacement=\&] {%
      \& \node[align=center] (t) {\scriptsize\ensuremath{#4}}; \& \\%
      \node (l) {\scriptsize\ensuremath{#1}}; \& \& \node (r) {\scriptsize\ensuremath{#3}}; \\%
      \& \node (b) {\scriptsize\ensuremath{#2}}; \& \\%
    };%
    \draw[gray] (l.north) |- (t.mid west);%
    \draw[gray] (l.south) |- (b.mid west);%
    \draw[gray] (t.mid east) -| (r.north);%
    \draw[gray] (b.mid east) -| (r.south);%
  \end{tikzpicture}%
}

\DeclareMathOperator{\dom}{dom}
\DeclareMathOperator{\im}{im}
\DeclareMathOperator{\Stab}{Stab}
\DeclareMathOperator{\Sub}{Sub}
\DeclareMathOperator{\Aff}{Aff}
\DeclareMathOperator{\SAff}{SAff}
\newcommand{\revbin}{\overleftarrow{\operatorname{bin}}}
\newcommand{\wt}{\widetilde}
\DeclareMathOperator{\idGrp}{\mathbbm{1}}
\newcommand*{\DecProblem}[1]{\textbf{\textsc{#1}}}

\newcommand*{\id}{\operatorname{id}}

\newcommand{\problem}[3][]{%
  \par\vspace{0.125cm plus 0.05cm minus 0.05cm}\begin{tabularx}{\textwidth-2\parindent}{lX}%
    \if\relax\detokenize{#1}\relax%
    \else%
      \textnormal{\textbf{Constant:}}&#1\\%
    \fi%
    \textnormal{\textbf{Input:}}&#2\\%
    \textnormal{\textbf{Question:}}&#3\\%
  \end{tabularx}\vspace{0.125cm plus 0.05cm minus 0.05cm}\par%
  }

\def\vlongrightarrow{\relbar\joinrel\longrightarrow}

\def\longmapright#1{\smash{\mathop{\vlongrightarrow}\limits^{#1}}}

\usepackage[affil-it]{authblk}

\author{Daniele D'Angeli\thanks{The first author was supported by the Austrian Science Fund projects FWF P24028-N18 and FWF P29355-N35.}}
\affil{Institut für Diskrete Mathematik\\
  Technische Universität Graz\\
  Steyrergasse 30\\
  8010 Graz, Austria}
\author{Emanuele Rodaro\thanks{The second author thanks the project INDAM-GNSAGA.}}
\affil{Department of Mathematics\\
  Politecnico di Milano\\
  Piazza Leonardo da Vinci, 32\\
  20133 Milano, Italy}
\author{Jan Philipp Wächter}
\affil{Institut für Formale Methoden der Informatik (FMI)\\
  Universität Stuttgart\\
  Universitätsstraße 38\\
  70569 Stuttgart, Germany}

\title{Automaton Semigroups and Groups: On the Undecidability of Problems Related to Freeness and Finiteness}

\begin{document}
  \maketitle
  \vspace{-\baselineskip}

\begin{abstract}
  \textbf{Abstract.}
  In this paper, we study algorithmic problems for automaton semigroups and automaton groups related to freeness and finiteness. In the course of this study, we also exhibit some connections between the algebraic structure of automaton (semi)groups and their dynamics on the boundary.

  First, we show that it is undecidable to check whether the group generated by a given invertible automaton has a positive relation, i.\,e.\ a relation $\bm{p} = \idGrp$ such that $\bm{p}$ only contains positive generators. Besides its obvious relation to the freeness of the group, the absence of positive relations has previously been studied by the first two authors and is connected to the triviality of some stabilizers of the boundary. We show that the emptiness of the set of positive relations is equivalent to the dynamical property that all (directed positive) orbital graphs centered at non-singular points are acyclic. Our approach also works to show undecidability of the freeness problem for automaton semigroups, which negatively solves an open problem by Grigorchuk, Nekrashevych and Sushchansky. In fact, we show undecidability of a strengthened version where the input automaton is complete and invertible.\marginline{Regarding the freeness problem, see the erratum below.}

  Gillibert showed that the finiteness problem for automaton semigroups is undecidable. In the second part of the paper, we show that this undecidability result also holds if the input is restricted to be bi-reversible and invertible (but, in general, not complete). As an immediate consequence, we obtain that the finiteness problem for automaton subsemigroups of semigroups generated by invertible, yet partial automata, so called \emph{automaton-inverse} semigroups, is also undecidable.
\end{abstract}

  \begin{section}{Introduction}
    Automaton groups, i.\,e.\ groups generated by synchronous Mealy automata, are a quite intriguing class of groups. They have deep connections with many areas of mathematics, from the theory of profinite groups to complex dynamics and theoretical computer science, and they serve as a source of examples or counterexamples for many important group theoretic problems (see e.\,g.\ \cite{zuk2012automata} for an introduction). Despite these connections and the many surprising and interesting consequences, knowledge about the class of automaton groups from the algebraic, algorithmic and dynamical perspective is still rather limited. From the algorithmic point of view, the word problem for automaton groups is decidable\footnote{See \cite{DAngeli2017} for some discussion of its complexity.} while many other problems are suspected (and sometimes proven) to be undecidable. In this regard, the most studied problems in the literature are the finiteness problem, the freeness problem and the conjugacy problem. The latter has been proven to be undecidable by {\v S}uni\'c  and Ventura \cite{Su-Ve09}. However, if this problem is restricted to the contracting case, then it turns out to be decidable, see \cite{BoBoSi13}. Decidability of the finiteness problem for automaton groups is still open. However, some partial results are known when certain properties of the generating automaton are relaxed: Belk and Bleak showed undecidability of the finiteness problem for groups generated by asynchronous automata \cite{belk2017some} and Gillibert showed undecidability of the finiteness problem for automaton semigroups \cite{Gilbert13}, which are generated by synchronous but not necessarily invertible Mealy automata. Other results worth mentioning in this respect are the recent proof by Gillibert showing undecidability of the order problem of automaton groups \cite{gillibert2017automaton} and the result by Bartholdi and Mitrofanov that the order problem is already undecidable for contracting automaton groups \cite{bartholdi2017wordAndOrderProblems}. On the other side, Klimann showed that the finiteness problem is solvable for reversible, invertible automata with two states or two letters and that so is the freeness problem for semigroups generated by two state invertible, reversible automata \cite{InesKli}. Obviously, it is usually easier to show undecidability results for automaton semigroups than it is for automaton groups. This might be one of the reasons why the less studied class of automaton semigroups seems to stir up more interest lately. There is, for example, the semigroup theoretic work of Cain \cite{cain2009automaton}, and Brough and Cain \cite{brough2015automaton, brough2017automatonTCS}, approaches to semigroups via duals of automata generating groups \cite{DaRo14, DaRo16} or studies of the torsion-freeness of automaton semigroups (see e.\,g.\ \cite{GKP}; also for some further references). Other works study the existence of free subsemigroups in automaton groups and semigroups (see e.\,g.\ \cite{DaRo13, DaRo14, DaRo16, francoeur2018existence, olukoya2017growthRates, SilSte05}).

    The aim of this paper is to make a contribution to the algorithmic point of view, in particular to the finiteness problem and the freeness problem for automaton (semi)groups. We approach these problems by making use of ideas with a dynamical flavor \cite{GriNeShu} underlining the importance of the dynamics on the boundary when considering algorithms: as an alternative to the automaton definition, one can also view automaton (semi)groups as level-preserving (finite-state) actions on the rooted infinite tree $\Sigma^*$ corresponding to the finite words over some alphabet $\Sigma$; the limit of this tree (in the sense of Gromov-Hausdorff convergence) is $\Sigma^\omega$, which corresponds to the right-infinite words over $\Sigma$ and on which the group also acts. As automaton (semi)groups are finitely generated and, thus, countable, their action on $\Sigma^\omega$ cannot be transitive. In this context, one of the main problems in the combinatorial study of automaton groups is to understand the (boundary) orbital Schreier graphs of a given automaton group \cite{MR2643891}. It is natural to ask which properties of a group can be recognized by exploring the structure of these Schreier graphs. Of course, this idea can be extended to semigroups by considering their orbital graphs on the boundary.

    After introducing some preliminaries, we first consider the freeness problem for automaton (semi)groups. Inspired by links between the freeness of an automaton group and the dynamics on periodic points in the boundary of the dual automaton exhibited in \cite{DaRo14, DaRo16}, we discuss the problem of checking whether a given automaton generating a group yields a positive relation, i.\,e.\ a relation $\bm{q} = \idGrp$ where $\bm{q}$ contains only positive generators, which was previously studied in \cite{DaGoKPRo16}. In \autoref{prop: characterization empty}, we first show that the existence of a positive relation is equivalent to having a cycle in some boundary Schreier graph. Then, we proceed to show undecidability of the existence of a positive relation; a problem left open in \cite[Problem 3]{DaGoKPRo16}. We do this by reducing \DecProblem{ICP}, an undecidable \cite[Theorem 11]{Pota-Bell} problem similar to Post's Correspondence Problem, to the problem using a construction given in \cite{Su-Ve09}. We observe that, using the same approach, one can reduce \DecProblem{IICP}, a variation of \DecProblem{ICP} of unknown decidability, to the freeness problem of automaton groups. Furthermore, we show using a similar reduction that (even a strengthened version of) the freeness problem for automaton semigroups is undecidable in \autoref{theo: freeness semigroup} and \autoref{cor:SemigroupFreenessIsUndecidable}.

    In the second part of the paper, we deal with the finiteness problem for automaton groups. In the spirit of \cite{Gilbert13} and \cite{DaGoKPRo16}, we attack the problem using Wang tilings. First, we show that finiteness of an automaton (semi)group is characterized by uniform boundedness of the orbital graphs in the boundary in \autoref{prop:uniformly bounded}. Afterwards, in \autoref{prop:recurrentTiling}, we connect the finiteness of a semigroup generated by a (partial) automaton obtained from a Wang tile set with the existence of non-periodic tiling, which is the crucial point of reducing the tiling problem for 4-way deterministic tiles (proved undecidable by Lukkarila \cite{Lukkarila}) to a strengthened version of the finiteness problem for automaton semigroups where the input automaton is known to be bi-reversible and invertible, which is stated to be undecidable in \autoref{theo: undecidability finite}. The strengthened version is already quite close to the group case; the only missing part is completeness of the automaton. An immediate consequence is that it is undecidable to check whether a subsemigroup of an automaton-inverse semigroup (i.\,e.\ the inverse semigroup generated by an invertible, yet partial automaton) is finite.
\end{section}

  \begin{section}{Preliminaries}\label{sec: preliminaries}
    \paragraph{Fundamentals.}
    We use $A \sqcup B$ to denote the disjoint union of two sets $A$ and $B$. The set $\mathbb{N}$ is the set of natural numbers including $0$, $\mathbb{Z}$ is the set of integers and $\mathbb{Z}_+$ is the set of (strictly) positive integers. Finally, $\mathbb{R}$ is the set of real numbers.

    An \emph{alphabet} $\Sigma$ is a non-empty finite set of \emph{letters}. By $\Sigma^*$, we denote the set of all finite words over the alphabet $\Sigma$, including the empty word, which we denote by $\varepsilon$. Furthermore, we set $\Sigma^+ = \Sigma^* \setminus \{ \varepsilon \}$ and use the notation $\Sigma^n$ for the set of words of length $n$ over $\Sigma$ and the notation $\Sigma^{\leq n}$ for the set of words $w$ of length up to $|w| = n$; $\Sigma^{< n}$ is the set of words of length smaller than $n$, respectively. The set of $\omega$-words, i.\,e.\ infinite words whose positions are the natural numbers, is denoted by $\Sigma^\omega$. Whenever we refer to a \emph{word} in this paper, it can either be a finite word or an $\omega$-word. The set $\Sigma^* \cup \Sigma^\omega$ is a metric space and, thus, also a topological space. The distance between two distinct words $x$ and $y$ is $|\Sigma|^{-|p|}$ where $p$ is the longest common prefix\footnote{A word $w$ is a \emph{prefix} of a word $x$ if there is a word $y$ with $wy = x$.} of $x$ and $y$. Notice that this metric induces the discrete topology on $\Sigma^*$, $\Sigma^n$ and $\Sigma^{\leq n}$ for all $n \in \mathbb{N}$. On $\Sigma^\omega = \prod_{i \in \mathbb{N}} \Sigma^1$, it induces Tychonoff's topology. The sequence $\left( \Sigma^{\leq n} \right)_{n \in \mathbb{N}}$ of compact metric spaces converges to $\Sigma^* \cup \Sigma^\omega$ (in the sense of Gromov-Hausdorff\footnote{For Gromov-Hausdorff convergence, see \cite{gromov2007metric}.}); therefore, $\Sigma^\omega$ is often referred to as the \emph{boundary} of $\Sigma^*$.

    On the algebraic side, we will be dealing with semigroups (and monoids), inverse semigroups and groups. In this context, we want to remind the reader of the difference between inverses in semigroups and in groups. An element $\inverse{s}$ of a semigroup $S$ is called (semigroup) \emph{inverse} to another element $s \in S$ if $\inverse{s} s \inverse{s} = \inverse{s}$ and $s \inverse{s} s = s$ hold. On the other hand, we will call an element $m^{-1}$ of a monoid $M$ (which may be a group) the \emph{group inverse} of $m \in M$ if we have $m^{-1} m = m m^{-1} = \idGrp$. Here and throughout this paper, we use $\idGrp$ to denote the neutral element of a monoid (or group, of course); the monoid in question will be clear from the context. Clearly, the group inverse of a monoid element is always also a (semigroup) inverse. The converse does not hold in general, however. To emphasize this difference, we use the notation $\inverse{m}$ to denote (semigroup) inverses and the notation $m^{-1}$ to denote group inverses.

    A semigroup $S$ is an \emph{inverse semigroup} if every element $s \in S$ has a unique inverse $\inverse{s}$, see \cite{howie, petrich1984} for further details on inverse semigroups. As the Preston-Vagner Theorem (see \cite[p.~150]{howie} or \cite[p.~168]{petrich1984}) demonstrates, there is a close connection between inverse semigroups and \emph{partial function} (or \emph{partial maps}). For these, we fix some notation. To indicate that a function $f$ from a set $A$ to a set $B$ is partial, we write $f: A \to_p B$. The \emph{domain} of $f$, denoted by $\dom f$, is the subset of $A$ on which $f$ is defined. If we have $\dom f = A$, i.\,e.\ that $f$ is defined on all elements in $A$, then we call $f$ a \emph{total} function from $A$ to $B$ and write $f: A \to B$. The counter part to $\dom f$ is $\im f$, the \emph{image} of $f$; it consists of the images under $f$ of the elements in $\dom f$. We say that a partial functions $\inverse{f}: B \to_p A$ is \emph{inverse} to another partial function $f: A \to_p B$ if $\dom \inverse{f} = \im f$, $\im \inverse{f} = \dom f$ and $f(\inverse{f}(f(a))) = f(a)$ for all $a \in \dom f$ as well as $\inverse{f}(f(\inverse{f}(b))) = \inverse{f}(b)$ for all $b \in \im f$; in an abuse of terminology, we say that $f^{-1} = \inverse{f}$ is a \emph{group inverse} of $f$ if $\dom f^{-1} = \im f$, $\im f^{-1} = \dom f$ and $f^{-1}(f(a)) = a$ for all $a \in \dom f$ as well as $f(f^{-1}(b)) = b$ for all $b \in \im f$. Note that in the latter case both functions are injective and the group inverse is unique.

    \paragraph*{Automata: Definition, Properties and Operations.}
    In this paper, the term \emph{automaton} refers to a special form of a finite-state transducer: it is
    \begin{itemize}[noitemsep]
      \item synchronous (it outputs exactly one letter on input of one letter),
      \item not necessarily complete and
      \item its input and output alphabets coincide.
    \end{itemize}

    Formally, let $Q$ be a non-empty set and let $\Sigma$ be an alphabet. An automaton is a triple $\mathcal{T} = (Q, \Sigma, \delta)$ with $\delta \subseteq Q \times \Sigma \times \Sigma \times Q$. The set $Q$ is called the \emph{state set} of $\mathcal{T}$, $\Sigma$ is its (input and output) \emph{alphabet} and $\delta$ is its \emph{transition set}. Accordingly, an element $(q, a, b, p) \in \delta$ is called a \emph{transition} of the automaton from state $q$ on \emph{input} $a$ with \emph{output} $b$ into state $p$. To denote a transition, we use the more graphical notation $\trans{q}{a}{b}{p}$ instead of the tuple notation $(q, a, b, p)$. Additionally, we use the common graphical representation for automata: a transition $\trans{q}{a}{b}{p} \in \delta$ is represented as
    \begin{center}
      \begin{tikzpicture}[baseline=(q.base), auto, shorten >=1pt, >=latex]
        \node[state] (q) {$q$};
        \node[state, right=of q] (p) {$p$};
        \path[->] (q) edge node {$a/b$} (p);
      \end{tikzpicture}.
    \end{center}

    An automaton $\mathcal{T} = (Q, \Sigma, \delta)$ is called \emph{deterministic} if
    \[
      \left| \left\{ \trans{q}{a}{b}{p} \mid \trans{q}{a}{b}{p} \in \delta, b \in \Sigma, p \in Q \right\} \right| \leq 1
    \]
    holds for all $q \in Q$ and all $a \in \Sigma$. It is called \emph{complete} if we have
    \[
      \left| \left\{ \trans{q}{a}{b}{p} \mid \trans{q}{a}{b}{p} \in \delta, b \in \Sigma, p \in Q \right\} \right| \geq 1
    \]
    for all $q \in Q$ and all $a \in \Sigma$. An automaton $\mathcal{T} = (Q, \Sigma, \delta)$ is \emph{reversible} if
    \[
      \left| \left\{ \trans{q}{a}{b}{p} \mid \trans{q}{a}{b}{p} \in \delta, b \in \Sigma, q \in Q \right\} \right| \leq 1
    \]
    holds for all $a \in \Sigma$ and $p \in Q$ (i.\,e.\ if it is co-deterministic with respect to the input). This means that $\trans{q}{a}{b}{p}, \trans{q'}{a}{b'}{p} \in \delta$ implies $q = q'$ and $b = b'$ for every $a, b, b' \in \Sigma$ and every $q, q', p \in Q$.

    For every automaton $\mathcal{T} = (Q, \Sigma, \delta)$, one can define its \emph{inverse} automaton $\inverse{\mathcal{T}} = (\inverse{Q}, \Sigma, \inverse{\delta})$ where $\inverse{Q}$ is a disjoint copy of $Q$ and we have
    \[
      \trans{\inverse{q}}{b}{a}{\inverse{p}} \in \inverse{\delta} \iff \trans{q}{a}{b}{p} \in \delta \text{,}
    \]
    i.\,e.\ we swap input and output. By defining $\inverse{\inverse{q}} = q$ for all $q \in Q$, we obtain that taking the inverse of an automaton is an involution, i.\,e.\ we have $\inverse{\inverse{\mathcal{T}}} = \mathcal{T}$.
    For any of the automaton properties defined above, we also have an inverse version: we say $\mathcal{T}$ is \emph{inverse-deterministic} (sometimes also called \emph{invertible}) if $\inverse{\mathcal{T}}$ is deterministic, it is \emph{inverse-complete} if $\inverse{\mathcal{T}}$ is complete, and it is \emph{inverse-reversible} if $\inverse{\mathcal{T}}$ is reversible (this is the case if $\mathcal{T}$ is co-deterministic with respect to the output). Additionally, we define versions of the properties which describe that they hold for an automaton and its inverse at the same time: an automaton $\mathcal{T}$ is \emph{bi-deterministic} if it is deterministic and inverse-deterministic, it is \emph{bi-complete} if it is complete and inverse-complete and it is \emph{bi-reversible} if it is reversible and inverse-reversible.

    Next to the inverse of an automaton, there is also its \emph{dual}. For an automaton $\mathcal{T} = (Q, \Sigma, \delta)$, we define its dual $\partial \mathcal{T} = (\Sigma, Q, \partial \delta)$ by
    \[
      \trans{a}{q}{p}{b} \in \partial \delta \iff \trans{q}{a}{b}{p} \in \delta \text{,}
    \]
    i.\,e.\ we swap the roles of states and letters. Just like taking the inverse, taking the dual is an involution: $\partial \partial \mathcal{T} = \mathcal{T}$. Notice that there are many connections between $\mathcal{T}$ and its dual. For example, we have that
    \begin{itemize}
      \item $\mathcal{T}$ is deterministic if and only if $\partial \mathcal{T}$ is deterministic,
      \item $\mathcal{T}$ is complete if and only if $\partial\mathcal{T}$ is complete,
      \item $\mathcal{T}$ is inverse-deterministic if and only if $\partial\mathcal{T}$ is reversible, and that
      \item $\mathcal{T}$ is inverse-reversible if and only if $\partial\mathcal{T}$ is inverse-reversible.
    \end{itemize}

    Other operations on automata involve two (or more) automata. For example, for any two automata $\mathcal{T}_1 = (Q, \Sigma, \delta)$ and $\mathcal{T}_2 = (P, \Gamma, \rho)$, one can take their \emph{union} automaton $\mathcal{T}_1 \cup \mathcal{T}_2 = (Q \cup P, \Sigma \cup \Gamma, \delta \cup \rho)$. For the union of an automaton $\mathcal{T} = (Q, \Sigma, \delta)$ and its inverse $\inverse{\mathcal{T}} = (\inverse{Q}, \Sigma, \inverse{\delta})$, we use a shorthand notation and denote it by $\wt{\mathcal{T}} = (\wt{Q}, \Sigma, \wt{\delta})$, where $\wt{Q} = Q \cup \inverse{Q}$ and $\wt{\delta} = \delta \cup \inverse{\delta}$.

    In addition to the union, one can also take the \emph{composition} of two automata $\mathcal{T}_1 = (Q, \Sigma, \delta)$ and $\mathcal{T}_2 = (P, \Sigma, \rho)$. It is the automaton $\mathcal{T}_2 \circ \mathcal{T}_1 = (P \circ Q, \Sigma, \rho \circ \delta )$ with $P \circ Q = \{ p \circ q \mid p \in P, q \in Q \}$ (a formal copy of $P \times Q$) given by
    \[
      \trans{p \circ q}{a}{c}{p' \circ q'} \in \rho \circ \delta \iff \trans{q}{a}{b}{q'} \in \delta, \trans{p}{b}{c}{p'} \in \rho \text{ for some $b \in \Sigma$} \text{.}
    \]
    The idea is to use the output $b$ of $\mathcal{T}_1$ on input of $a$ as input for the second automaton $\mathcal{T}_2$. Notice that composition preserves the properties defined above: if $\mathcal{T}_1$ and $\mathcal{T}_2$ are deterministic/complete/reversible, then so is $\mathcal{T}_1 \circ \mathcal{T}_2$.

    A special form of the composition of automata, is the \emph{$k^{\text{th}}$-power} $\mathcal{T}^k$ of an automaton $\mathcal{T}=(Q, \Sigma, \delta)$ for some $k \geq 1$. It is the $k$-fold composition of $\mathcal{T}$ with itself:
    \[
      \mathcal{T}^k = \underbrace{\mathcal{T} \circ \mathcal{T} \circ \dots \circ \mathcal{T}}_{k \text{ times}} \text{.}
    \]

    \paragraph*{Semantics of Automata, Automaton Semigroups.}
    So far, we have defined automata only formally. Of course, behind these definitions stands the common intuitive understanding that an automaton emits an output for some input. Since we will be dealing primarily with deterministic automata in this paper, we give this class a special name: a deterministic automaton is called an \emph{$S$-automaton} from now on. This name stems from the fact that these automata generate semigroups (as we will see shortly). Later on, we will also encounter $\inverse{S}$-automata (generating inverse semigroups) and $G$-automata (generating groups). Every $S$-automaton $\mathcal{T} = (Q, \Sigma, \delta)$ induces a partial right action $\cdot: \Sigma \times Q \to_p Q$ of $\Sigma$ on $Q$: we have $q \cdot a = p$ if $\trans{q}{a}{b}{p} \in \delta$ for some $b \in \Sigma$; here, we use the more common infix notation for $\cdot$. Notice that $b$ must be unique since $\mathcal{T}$ is deterministic by definition. This action can be extended into a partial action of $\Sigma^*$ on $Q$ by setting $q \cdot \varepsilon = q$ and $q \cdot a_1 a_2 \dots a_n = \left( \left( q \cdot a_1 \right) \cdot a_2 \right) \ldots \cdot a_n$ (if defined). Intuitively, this action describes reading the finite \emph{input} word in the automaton: if one starts reading the finite word $w$ in state $q$, one ends up in $q \cdot w$. Notice that this action is total, i.\,e.\ $\cdot$ is a function $\Sigma^* \times Q \to Q$, if (and only if) $\mathcal{T}$ is complete.

    Additionally, every $S$-automaton $\mathcal{T}$ induces a partial left action $\circ: Q \times \Sigma \to_p \Sigma$ of $Q$ on $\Sigma$: we have $q \circ a = b$ if $\trans{q}{a}{b}{p} \in \delta$ for some $p \in Q$; again, we use infix notation for $\circ$. This action can be extended into an action of $Q$ on $\Sigma^*$ by setting $q \circ \varepsilon = \varepsilon$ and $q \circ a w = (q \circ a) \left[ (q \cdot a) \circ w \right]$ (if defined) for all $a \in \Sigma$ and $w \in \Sigma^+$. Here, the intuition is to read a finite input word $a_1 \dots a_n$ in the automaton starting in state $q$. After each letter, the automaton emits an output letter and transitions into a new state. From this state onwards, the next letter is read and so on. Notice that this action is also total, i.\,e.\ $\circ$ is a function $Q \times \Sigma^* \to \Sigma^*$, if (and only if) $\mathcal{T}$ is complete. Furthermore, notice that $\circ$ can be extended into a (partial) function $Q \times \Sigma^\omega \to_p \Sigma^\omega$.

    We can also consider the action of each individual state $q$. We re-use the notation $q \circ{}\!$ to denote the function $q \circ{}\! : \Sigma^* \cup \Sigma^\omega \to_p \Sigma^* \cup \Sigma^\omega$ induced by $\circ$ with first parameter $q$. Notice that, because the automaton is synchronous, all functions $q \circ{}\!$ are length-preserving (whenever they are defined) and prefix-compatible. Now, we can consider the closure of the functions $Q \circ{}\! = \{ q \circ{}\! \mid q \in Q \}$ for some $S$-automaton $\mathcal{T} = (Q, \Sigma, \delta)$ under (finite) composition of partial functions. This is the \emph{semigroup generated by $\mathcal{T}$}, which we denote by $\mathscr{S}(\mathcal{T})$. A semigroup is called an \emph{automaton semigroup} if it is generated by some $S$-automaton $\mathcal{T}$. To avoid notational overhead, we omit the $\circ$ symbols in an element $\bm{q} \circ{}\! = q_n \circ \dots \circ q_1 \circ{}\!$ of an automaton semigroup and simply write $\bm{q} \circ{}\! = q_n \dots q_1 \circ{}\!$.

    Just like we extended the notation $q \circ u$ to cover more than a single state, we can do the same with the notation $q \cdot u$: for an $S$-automaton $\mathcal{T} = (Q, \Sigma, \delta)$, a finite word $u \in \Sigma^*$ and states $q_n, \dots, q_2, q_1 \in Q$, define $q_n \dots q_2 q_1 \cdot u = \left[ q_n \dots q_2 \cdot (q_1 \circ u) \right] (q_1 \cdot u)$ inductively; furthermore, we define $\varepsilon \cdot u = \varepsilon$. Notice that $q_n \dots q_2 q_1 \cdot u$ with this definition coincides with the state reached in $\mathcal{T}^n$ if one starts reading the finite input word $u$ in state $q_n \circ \dots \circ q_2 \circ q_1$.

    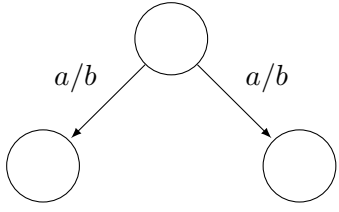
\begin{figure}
      \centering
      \begin{tikzpicture}[auto, shorten >=1pt, >=latex]
        \node[state] (q) {};
        \node[state, below left=of q] (qa) {};
        \node[state, below right=of q] (qc) {};

        \path[->] (q) edge node[swap] {$a/b$} (qa)
                  (q) edge node {$a/b$} (qc);
      \end{tikzpicture}%
      \caption{A bi-re\-ver\-si\-ble automaton which is neither deterministic nor inverse-deterministic}\label{fig:bireversibleNotDetNorInvDet}%
    \end{figure}%
    \begin{remark}
      Notice that the definition of an automaton semigroup presented here differs from the more common one, which is based on \emph{complete} $S$-automata.\footnote{This extended definition was previously used in \cite{DAngeli2017}.}. We call an automaton semigroup which is generated by a complete $S$-automaton a \emph{complete automaton semigroup} to distinguish the two concepts. There are obviously some connections between the two classes but it is not clear whether they coincide or not (see \cite[Section~3]{structurePart} for a discussion of this). The usual way of going from automaton semigroups to complete automaton semigroups is to algebraically adjoin a zero element. For a semigroup $S$, let $S^0$ denote the semigroup obtained from $S$ by adjoining a new zero element. With this notation, if $S$ is an automaton semigroup, then $S^0$ is a complete automaton semigroup \cite[Proposition 1]{DAngeli2017} (but also note the discussion in \cite[Section~3]{structurePart}).

      Notice that non-complete automata sometimes behave differently compared to complete automata. For example, for them, reversibility does not imply determinism as can be seen in \autoref{fig:bireversibleNotDetNorInvDet}.
    \end{remark}
    
    \paragraph*{Inverse Automaton Semigroups and Automaton Groups.}
    Using non-complete automata to define automaton semigroups allows us to define automaton-inverse semigroups, an intermediate step between automaton semigroups and automaton groups.\footnote{It seems that the concept of inverse semigroups generated by partial automata has not been studied widely yet. However, it does appear in \cite{nekrashevych2006self}, for example, and a similar concept was studied by Olijnyk, Sushchansky and Slupik \cite{olijnyk2010inverse} (see also there for previous work by Sushchansky and Slupik).} To do so, we introduce the name \emph{$\inverse{S}$-automaton} to denote bi-deterministic automata. For an $\inverse{S}$-automaton $\mathcal{T} = (Q, \Sigma, \delta)$ and its inverse $\inverse{\mathcal{T}} = (\inverse{Q}, \Sigma, \inverse{\delta})$, we can consider the set $\wt{Q} \circ{}\! = (Q \circ{}\!) \cup (\inverse{Q} \circ{}\!) = \{ \wt{q} \circ{}\! \mid \wt{q} \in Q \cup \inverse{Q} = \wt{Q} \}$ of actions induced by the states of the automaton and its inverse. The \emph{inverse semigroup generated by $\mathcal{T}$}, denoted by $\inverse{\mathscr{S}}(\mathcal{T})$, is the closure of $\wt{Q} \circ{}\!$ under finite composition of partial functions. Notice that it is equal to the semigroup generated by the disjoint union of $\mathcal{T}$ and its inverse $\inverse{\mathcal{T}}$, i.\,e.\ we have $\inverse{\mathscr{S}}(\mathcal{T}) = \mathscr{S}(\mathcal{T} \sqcup \inverse{\mathcal{T}})$. A semigroup is called an \emph{automaton-inverse semigroup} if it is generated by some $\inverse{S}$-automaton. The name comes from the fact that $\inverse{q} \circ{}\!$ and $q \circ{}\!$ are mutually inverse in the sense of partial functions for all states $q \in Q$: $q \circ \inverse{q} \circ q \circ{}\! = q \circ{}\!$ and $\inverse{q} \circ q \circ \inverse{q} \circ{}\! = \inverse{q} \circ{}\!$. Notice that, therefore, $\inverse{\mathscr{S}}(\mathcal{T})$ is an inverse semigroup for all $\inverse{S}$-automata $\mathcal{T}$. Please note that the use of $\inverse{s}$ to denote the inverse of $s$ in a semigroup is compatible with the notation $\inverse{q}$ for the corresponding state in the inverse automaton. Also note the difference in definition between an automaton-inverse semigroup and an inverse automaton semigroup!\footnote{However, it turns out that both concepts coincide \cite[Theorem~25]{structurePart}.}

    \begin{example}
      The (finite) Brandt semigroup $B_2$ is generated by the elements $p$ and $q$ with the relations $p^2 = q^2 = 0$, $pqp = p$ and $qpq = q$; it, thus, contains the elements $\{ p, q, pq, qp, 0 \}$ \cite[p.~32]{howie}\footnote{Readers familiar with syntactic semigroups might also find it interesting that $B_2$ is the syntactic semigroup of $\{ (pq)^n \mid n \geq 1\}$.}. To realize this semigroup as an automaton semigroup, we can let it act on itself (see the proof of \cite[Proposition 4.6]{cain2009automaton}). This leads to the $S$-automaton $\mathcal{T} = (\{ q, p \}, \{ a, b, ab, ba, 0 \}, \delta)$ depicted in \autoref{fig:BrandtSemigroupAutomaton}.
      \begin{figure}
        \begin{center}
          \begin{tikzpicture}[auto, shorten >=1pt, >=latex]
            \node[state] (q) {$p$};
            \node[state, right=of q] (p) {$q$};
  
            \path[->] (q) edge[loop left] node {
              $\begin{aligned}
                a &/ 0 \\
                b &/ ab \\
                ab &/ 0 \\
                ba &/ a \\
                0 &/ 0
              \end{aligned}$} (q);
            \path[->] (p) edge[loop right] node {
              $\begin{aligned}
                a &/ ba \\
                b &/ 0 \\
                ab &/ b \\
                ba &/ 0 \\
                0 &/ 0
              \end{aligned}$} (p);
          \end{tikzpicture}
        \end{center}
        \caption{An automaton generating the Brandt semigroup $B_2$.}\label{fig:BrandtSemigroupAutomaton}
      \end{figure}
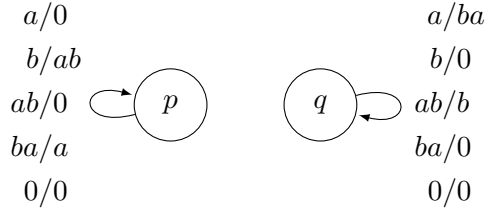
      Here, we have used the alphabet $\{ a, b, ab, ba, 0 \}$ (where we consider $ab$ and $ba$ to be \emph{single} letters) instead of the elements of the semigroup to have a clearer distinction between the two concepts; the idea is to interpret $p$ as $a$ and $q$ as $b$. While the Brandt semigroup $\mathscr{S}(\mathcal{T})$ is an inverse semigroup with $\inverse{q} = p$ and $\inverse{p} = q$, the automaton $\mathcal{T}$ is \emph{not} an $\inverse{S}$-automaton as it is not inverse-deterministic. Therefore, $\inverse{\mathscr{S}}(\mathcal{T})$ has no defined meaning.
    \end{example}

    Finally, we call a complete $\inverse{S}$-automaton a \emph{$G$-automaton}. Notice that, by reasons of cardinality, any $G$-automaton is not only complete but also bi-complete. Thus, for a $G$-automaton $\mathcal{T} = (Q, \Sigma, \delta)$ with inverse $\inverse{\mathcal{T}} = (Q, \Sigma, \delta)$, all functions $q \circ{}\!$ and $\inverse{q} \circ{}\!$ with $q \in Q$ are total length-preserving bijections $\Sigma^* \to \Sigma^*$ (or $\Sigma^\omega \to \Sigma^\omega$, respectively). In this case, the automaton-inverse semigroup $\inverse{\mathscr{S}}(\mathcal{T})$ generated by $\mathcal{T}$ (seen as an $\inverse{S}$-automaton) is a group. We call it the \emph{group generated by $\mathcal{T}$} and denote it by $\mathscr{G}(\mathcal{T})$ to emphasize this fact. We call a group an \emph{automaton group} if it is generated by some $G$-automaton $\mathcal{T}$.

    We also use the (less-precise) name of an \emph{automaton structure} to denote an automaton semigroup, an automaton-inverse semigroup or an automaton group. We summarize the definitions from above in \autoref{tab:automatonStructures}.
    \begin{table}[h!]
      \begin{center}
        \begin{tabular}{lll}
          Automaton Class & Properties & Generated Structure\\\hline
          $S$-automaton & deterministic & semigroup $\mathscr{S}(\mathcal{T})$\\
          $\inverse{S}$-automaton & bi-deterministic & inverse semigroup $\inverse{\mathscr{S}}(\mathcal{T})$\\
          $G$-automaton & bi-deterministic, bi-complete & group $\mathscr{G}(\mathcal{T})$%
        \end{tabular}%
      \end{center}%
      \vspace{-\baselineskip}%
      \caption{Structures defined by automata}\label{tab:automatonStructures}
    \end{table}

    \begin{example}\label{expl:automatonForZ}
      The automaton
      \begin{center}
        \begin{tikzpicture}[auto, shorten >=1pt, >=latex]
          \node[state] (+1) {$+1$};
          \node[state, right=of +1] (+0) {$+0$};

          \path[->] (+1) edge[loop left] node {$1/0$} (+1)
                         edge node {$0/1$} (+0)
                    (+0) edge[loop right] node[align=left] {$0/0$\\$1/1$} (+0)
          ;
        \end{tikzpicture},
      \end{center}
      which shall be denoted by $\mathcal{T}$ for this example, is called the \emph{adding machine}. It is a $G$-automaton and, as such, it is also an $\inverse{S}$- and an $S$-automaton. To understand the group and the semigroup generated by $\mathcal{T}$, it is useful to study the actions of $+1$ and $+0$.

      The state $+0$ obviously acts as the identity on $\{ 0, 1 \}^* \cup \{ 0, 1 \}^\omega$. The action $+1 \circ{}\!$ is more interesting. We can see any finite word from $\{ 0, 1 \}^*$ as representing a binary number in reverse (i.\,e.\ with least significant bit first). The same is true for $\omega$-words of the form $w0^\omega$ for some $w \in \{ 0, 1 \}^*$ as the infinitely many $0$s at the end can be considered leading $0$s in the binary representation. Now, the action of $+1$ increments the encoded number by one. For example, we have $+1 \circ 010 = 110$ and $+1 \circ 1100^\omega = 0010^\omega$. Therefore, it is not surprising that $\mathscr{S}(\mathcal{T})$, the semigroup generated by $\mathcal{T}$, is the free monoid with one generator. Similarly, $\mathscr{G}(\mathcal{T})$ is the free group with one generator.

      If we extend\footnote{This extension is inspired by \cite[Fig.~8]{olijnyk2010inverse}.} the adding machine into the automaton
      \begin{center}
        \begin{tikzpicture}[auto, shorten >=1pt, >=latex]
          \node[state] (+1) {$+1$};
          \node[state, right=of +1] (+0) {$+0$};

          \path[->] (+1) edge[loop left] node {$1/0$} (+1)
                         edge node {$0/1$} (+0)
                         edge node[swap] {$\hat{0}/\hat{1}$} (+0)
                    (+0) edge[loop right] node[align=left] {$0/0$\quad$\hat{0}/\hat{0}$\\$1/1$\quad$\hat{1}/\hat{1}$} (+0)
          ;
        \end{tikzpicture},
      \end{center}
      which we denote by $\hat{\mathcal{T}}$, then the result is not complete anymore but it still is an $\inverse{S}$-automaton whose inverse is
      \begin{center}
        \begin{tikzpicture}[auto, shorten >=1pt, >=latex]
          \node[state] (+1) {$\inverse{+1}$};
          \node[state, right=of +1] (+0) {$\inverse{+0}$};

          \path[->] (+1) edge[loop left] node {$0/1$} (+1)
                         edge node {$1/0$} (+0)
                         edge node[swap] {$\hat{1}/\hat{0}$} (+0)
                    (+0) edge[loop right] node[align=left] {$0/0$\quad$\hat{0}/\hat{0}$\\$1/1$\quad$\hat{1}/\hat{1}$} (+0)
          ;
        \end{tikzpicture}.
      \end{center}
      Notice that we have $\inverse{+0} \circ{}\! = +0 \circ{}\!$. Similarly to our previous examples, $\inverse{\mathscr{S}}(\hat{\mathcal{T}})$, the inverse semigroup generated by $\hat{\mathcal{T}}$, is the free inverse monoid with one generator. By \cite[VIII.4.6, p.~381]{petrich1984} (see also \cite[Lemma~24]{olijnyk2010inverse}), this follows if we show $+1 \inverse{+1} (\inverse{+1})^n (+1)^n \neq (\inverse{+1})^n (+1)^n$ and $\inverse{+1}+1(+1)^n(\inverse{+1})^n \neq (+1)^n(\inverse{+1})^n$ for all $n \in \mathbb{N}$. We have $(+1)^n \circ 0^n \hat{0} = w \hat{0}$ where $w$ is the (reverse/least significant bit first) binary representation of the number $n$ with length $n$. Thus, we also have $(\inverse{+1})^n \circ w \hat{0} = 0^n \hat{0}$. On the other hand, $\inverse{+1} \circ 0^n \hat{0}$ is undefined. This shows $+1 \inverse{+1} (\inverse{+1})^n (+1)^n \circ{}\! \neq (\inverse{+1})^n (+1)^n \circ{}\!$. To show the other inequality, one can use the word $1^n\hat{1}$.
    \end{example}

    \paragraph*{Orbital/Schreier Graphs}
    Many properties of automaton (semi)groups can be studied by exploring their so-called orbital (Schreier) graphs. For an $S$-automaton $\mathcal{T} = (Q, \Sigma, \delta)$ and a word $x \in \Sigma^* \cup \Sigma^\omega$\, we denote the \emph{orbit of $x$} under the action of $\mathcal{T}$ by
    \[
      Q^* \circ x = \{ q_n \dots q_1 \circ x \mid q_1, \dots, q_n \in Q, n \in \mathbb{N} \} \text{.}
    \]
    On this set, we can define a natural graph structure: the \emph{orbital graph} of $\mathcal{T}$ \emph{centered at $x$}, denoted by $\mathcal{T} \circ x$, is a labeled directed graph with $Q^* \circ x$ as the node set and the edge set $\{ \transa{y}{q}{q \circ x} \mid y \in Q^* \circ x, q \in Q, q \circ{}\! \text{ defined on } y \}$. Notice that, for an $\omega$-word $\xi \in \Sigma^\omega$, the graph $\mathcal{T} \circ \xi$ is the limit in the sense of pointed Gromov-Hausdorff convergence of the sequence $\mathcal{T} \circ \xi_n$ where $\xi_n$ is the prefix of length $n$ of $\xi$.
    If $\mathcal{T} = (Q, \Sigma, \delta)$ is an $\inverse{S}$-automaton (or even a $G$-automaton), then we can extend these notions to include inverses. Let $\inverse{\mathcal{T}} = (\inverse{Q}, \Sigma, \inverse{\delta})$ be the inverse of $\mathcal{T}$, let $x \in \Sigma^* \cup \Sigma^\omega$ be a word, and let $\wt{Q} = Q \cup \inverse{Q}$. Then, we can define the orbit of $x$ under the action of $\mathcal{T}$ as an $\inverse{S}$-automaton:
    \[
      \wt{Q}^* \circ x = \{ \wt{q}_n \dots \wt{q}_1 \circ x \mid \wt{q}_1, \dots, \wt{q}_n \in \wt{Q}, n \in \mathbb{N} \} \text{.}
    \]
    Again, we have a natural graph structure on this set: the \emph{Schreier graph} of $\mathcal{T}$ \emph{centered at $x$} is the labeled directed graph $\wt{\mathcal{T}} \circ x$ with node set $\wt{Q}^* \circ x$. It contains an edge $y \longmapright{\wt{q}} z$ whenever $\wt{q} \circ y = z$ for $\wt{q} \in \wt{Q}$. Notice that $\mathcal{T} \circ x$ is always a sub-graph of $\wt{\mathcal{T}} \circ x$. Just like with $\mathcal{T} \circ \xi$, $\wt{\mathcal{T}} \circ \xi$ is the limit (in the sense of pointed Gromov-Hausdorff convergence) of the sequence $\wt{\mathcal{T}} \circ \xi_n$ where $\xi_n$ is the prefix of length $n$ of $\xi$.

    \begin{example}\label{expl:orbitalSchreierGraph}
      Recall the adding machine $\mathcal{T}$ from \autoref{expl:automatonForZ}. The following figure depicts the infinite orbital graph $\mathcal{T} \circ 0^\omega$ (dark) and the additional vertices and edges in the Schreier graph $\wt{\mathcal{T}} \circ 0^\omega$ \textcolor{gray!70}{(light)} schematically.
      \begin{center}
        \begin{tikzpicture}[auto, shorten >=1pt, >=latex,
            vertex/.style={circle, fill=black, draw=none, inner sep=1pt},
            light/.style={
              color=gray!70, prefix after command={\pgfextra{\tikzset{every label/.style={color=gray!70}}}}},
            node distance=2cm]
          \node[light] (ld) {$\dots$};
          \node[vertex, light, label=below:$101^\omega$, right=0cm of ld] (-3) {};
          \node[vertex, light, label=below:$011^\omega$, right=of -3] (-2) {};
          \node[vertex, light, label=below:$111^\omega$, right=of -2] (-1) {};
          \node[vertex, label=below:$000^\omega$, right=of -1] (0) {};
          \node[vertex, label=below:$100^\omega$, right=of  0] (1) {};
          \node[vertex, label=below:$010^\omega$, right=of  1] (2) {};
          \node[vertex, label=below:$110^\omega$, right=of  2] (3) {};
          \node[right=0cm of 3] {$\dots$};

          \path[->] (-3) edge[loop above, looseness=25, out=125, in=55, light] node {$+0$} (-3)
                    (-2) edge[loop above, looseness=25, out=125, in=55, light] node {$+0$} (-2)
                    (-1) edge[loop above, looseness=25, out=125, in=55, light] node {$+0$} (-1)
                     (0) edge[loop above, looseness=25, out=125, in=55] node {$+0$}  (0)
                     (1) edge[loop above, looseness=25, out=125, in=55] node {$+0$}  (1)
                     (2) edge[loop above, looseness=25, out=125, in=55] node {$+0$}  (2)
                     (3) edge[loop above, looseness=25, out=125, in=55] node {$+0$}  (3)
                    (-3) edge[light, bend left] node {$+1$} (-2)
                    (-2) edge[light, bend left] node {$+1$} (-1)
                    (-1) edge[light, bend left] node {$+1$}  (0)
                     (0) edge[bend left] node {$+1$}  (1)
                     (1) edge[bend left] node {$+1$}  (2)
                     (2) edge[bend left] node {$+1$}  (3)
                    (-2) edge[light] node {$\inverse{+1}$} (-3)
                    (-1) edge[light] node {$\inverse{+1}$} (-2)
                     (0) edge[light] node {$\inverse{+1}$} (-1)
                     (1) edge[light] node {$\inverse{+1}$} (0)
                     (2) edge[light] node {$\inverse{+1}$} (1)
                     (3) edge[light] node {$\inverse{+1}$} (2)
                    ;
        \end{tikzpicture}
      \end{center}
    \end{example}

    One may have noticed that $\widetilde{\mathcal{T}} \circ 0^\omega$ from \autoref{expl:orbitalSchreierGraph} coincides with the Cayley graph of $\mathscr{G}(\mathcal{T}) = (\mathbb{Z}, +)$. This is not a coincidence but the reason why $\widetilde{\mathcal{T}} \circ x$ is called the \emph{Schreier} graph: it is isomorphic to the coset graph of the stabilizer of $x$. For an $S$-automaton $\mathcal{T} = (Q, \Sigma, \delta)$ and a word $x \in \Sigma^* \cup \Sigma^\omega$, we define
    \[
      \Stab_{\mathcal{T}}(x) = \{ \bm{q} \mid \bm{q} \in Q^+, \bm{q} \circ x = x \}
    \]
    as the set of (positive length) state sequences whose actions stabilize $x$. Accordingly, $\Stab_{\mathcal{T}}(x) \circ{}\! = \{ \bm{q} \circ{}\! \mid \bm{q} \in Q^+, \bm{q} \circ x = x \} \subseteq \mathscr{S}(\mathcal{T})$ is the (semigroup) \emph{stabilizer of $x$} under the action induced by $\mathcal{T}$. Notice that, if $\mathcal{T}$ is an $\inverse{S}$-automaton, then, with this notation, we have
    \[
      \Stab_{\mathcal{T} \sqcup \inverse{\mathcal{T}}}(x) = \{ \bm{\wt{q}} \circ{}\! \mid \bm{\wt{q}} \in (Q \sqcup \inverse{Q})^+, \bm{\wt{q}} \circ x = x \} \text{.}
    \]
    If $\mathcal{T}$ is a $G$-automaton, then
    \[
      \Stab_{\mathcal{T} \sqcup \inverse{\mathcal{T}}}(x) \circ{}\! = \{ \bm{\wt{q}} \circ{}\! \mid \bm{\wt{q}} \in (Q \sqcup \inverse{Q})^*, \bm{\wt{q}} \circ x = x \} \subseteq \mathscr{G}(\mathcal{T})
    \]
    is the group stabilizer of $x$ under the action induced by $\mathcal{T}$. The isomorphism between $\widetilde{\mathcal{T}} \circ x$ and $\mathscr{G}/\Stab_{\mathcal{T} \sqcup \inverse{\mathcal{T}}}(x) \circ{}\!$ can now be seen easily (using the mapping $\bm{q} \circ{}\! \mapsto \bm{q} \Stab_{\mathcal{T} \sqcup \inverse{\mathcal{T}}}(x) \circ{}\!$). In addition to the name, this connection has another consequence: the size of $\widetilde{\mathcal{T}} \circ x$ is the index of $\Stab_{\mathcal{T} \sqcup \inverse{\mathcal{T}}}(x) \circ{}\!$ in $\mathscr{G}(\mathcal{T})$.

    The finiteness of orbital and Schreier graphs is an important property in studying certain algebraic properties of the corresponding automaton structure. This relationship has been exploited in \cite{DaRo16} in connection with the property of an automaton group to be free. In \cite{DaRo16} the authors implicitly used the following fact for automaton groups. It states that, when considering automaton groups, it is indifferent whether one considers the orbital graph or the Schreier graphs. The former is finite if and only if the latter is.

    \begin{lemma}\label{lem:groupAndSemigroupOrbitsCoincide}
      Let $\mathcal{T} = (Q, \Sigma, \delta)$ be a $G$-automaton and let $\inverse{\mathcal{T}} = (\inverse{Q}, \Sigma, \inverse{\delta})$ be its inverse. Then, for any $x \in \Sigma^* \cup \Sigma^\omega$, we have
      \[
        Q^* \circ x = \wt{Q}^* \circ x \quad \text{or} \quad |Q^* \circ x| = |\wt{Q}^* \circ x| = \infty
      \]
      where $\wt{Q} = Q \sqcup \inverse{Q}$.
    \end{lemma}
    \begin{proof}
      Let $x \in \Sigma^* \cup \Sigma^\omega$ be arbitrary. Trivially, we have $Q^* \circ x \subseteq \wt{Q}^* \circ x$. Thus, it remains to show $\wt{Q}^* \circ x \subseteq Q^* \circ x$ if $Q^* \circ x$ is finite (which is always the case for $x \in \Sigma^*$). Suppose, $Q^* \circ x \subsetneq \wt{Q}^* \circ x$. Then, there is an element $y \in \wt{Q}^* \circ x \setminus Q^* \circ x$ such that $z \longmapright{\inverse{q}} y$ is an edge in $\wt{\mathcal{T}} \circ x$ for some $z \in Q^* \circ x$ and some $\inverse{q} \in \inverse{Q}$. Because $Q^* \circ x$ is finite and because $\mathcal{T}$ is complete, $z$ must have an in-coming $q$-labeled edge from some  $y' \in Q^* \circ x$ by cardinality reasons. This yields the edge $\transa{z}{\inverse{q}}{y'}$ in $\wt{\mathcal{T}} \circ x$. However, this implies $y = y' \in Q^* \circ x$ (by the determinism of the graph); a contradiction.
    \end{proof}

    In the next lemma, we show that the completeness of the automaton is essential in the statement of \autoref{lem:groupAndSemigroupOrbitsCoincide}.
    \begin{lemma}\label{lem:inverseAndSemigroupOrbitsDoNotCoincide}
      Let $\mathcal{T} = (Q, \{ 0, \hat{0}, 1, \hat{1} \}, \delta)$ denote the following $\inverse{S}$-automaton (a kind of \enquote{mark adding} variation of the adding machine from \autoref{expl:automatonForZ})
      \begin{center}
        \begin{tikzpicture}[auto, shorten >=1pt, >=latex]
          \node[state] (q) {$q$};
          \node[state, right=of q] (p) {$p$};
    
          \path[->] (q) edge[loop left] node {$1/\hat{0}$} (q)
                        edge node {$0/\hat{1}$} (p)
                    (p) edge[loop right] node[align=left] {$0/\hat{0}$\\$1/\hat{1}$} (p)
          ;
        \end{tikzpicture}
      \end{center}
      and let $\inverse{\mathcal{T}} = (\inverse{Q}, \{ 0, \hat{0}, 1, \hat{1} \}, \inverse{\delta})$ denote its inverse
        \begin{center}
          \begin{tikzpicture}[auto, shorten >=1pt, >=latex]
            \node[state] (q) {$\inverse{q}$};
            \node[state, right=of q] (p) {$\inverse{p}$};
    
            \path[->] (q) edge[loop left] node {$\hat{0}/1$} (q)
                          edge node {$\hat{1}/0$} (p)
                      (p) edge[loop right] node[align=left] {$\hat{0}/0$\\$\hat{1}/1$} (p)
            ;
          \end{tikzpicture}.
        \end{center}
      Then, we have that $\inverse{\mathscr{S}}(\mathcal{T})$ is infinite and $|\wt{Q}^* \circ 0^\omega| = \infty$ where $\wt{Q} = Q \sqcup \inverse{Q}$ but also that $\mathscr{S}(\mathcal{T})$ is finite and $|Q^* \circ \xi| \leq 2$ for all $\xi \in \Sigma^\omega$ and, in particular $|Q^* \circ 0^\omega| = 2$.
    \end{lemma}
    \begin{proof}
      Any $\omega$-word from $\{ 0, 1 \}^\omega$ containing only finitely many occurrences of $1$ can be seen as the reverse binary representation of some natural number $n \in \mathbb{N}$ (with infinitely many leading/trailing zeros). For example, we have $0^\omega = \revbin(0)$ and $10110^\omega = \revbin(13)$, where $\revbin(n)$ denotes the (infinite) reverse binary representation of a natural number $n \in \mathbb{N}$.

      Recall from \autoref{expl:automatonForZ}, that the action of $q$ on such words basically increments the represented number by one. In the automaton depicted in the lemma, $q \circ{}$ additionally adds a decoration to each letter. Let $\langle \cdot \rangle: \{ 0, 1 \}^* \to \{ \hat{0}, \hat{1} \}^*$ denote the isomorphism given by $\langle 0 \rangle = \hat{0}$ and $\langle 1 \rangle = \hat{1}$. Then, we have $q \circ \revbin(n) = \langle \revbin(n + 1) \rangle$ and $p \circ{} = \langle \cdot \rangle$.

      Because $p \circ{}\!$ and $q \circ{}\!$ are only defined on words over $\{ 0, 1 \}$ (and, in particular, not on words containing at least one $\hat{0}$ or $\hat{1}$), we have, for all $\eta \in \Sigma^\omega$, that $Q^* \circ \eta$ contains $\eta$ itself and, possibly, a single other word (if $\eta$ does not contain any $\hat{0}$ or $\hat{1}$). For example, we have $Q^* \circ 0^\omega = \{ 0^\omega, \hat{1}\hat{0}^\omega \}$. Additionally, this shows that $\mathscr{S}(\mathcal{T})$ is finite.

      Since we can remove the decorations added by $q \circ{}$ using $\inverse{p}$ (if we take the inverses into consideration), we have $\revbin(n) = (\inverse{p} q)^n \circ 0^\omega \in \wt{Q}^* \circ 0^\omega$ for all $n \in \mathbb{N}$ and, thus, $|\wt{Q}^* \circ 0^\omega| = \infty$. This also shows that $\inverse{\mathscr{S}}(\mathcal{T})$ is infinite.
    \end{proof}
    \autoref{lem:inverseAndSemigroupOrbitsDoNotCoincide} has another consequence. For a $G$-automaton $\mathcal{T}$, we have that $\mathscr{G}(\mathcal{T})$ is finite if and only if $\mathscr{S}(\mathcal{T})$ is \cite{aklmp12}. However, the lemma states that the analog for $\inverse{S}$-automata does not hold: there is an $\inverse{S}$-automaton $\mathcal{T}$ such that $\mathscr{S}(\mathcal{T})$ is finite but $\inverse{\mathscr{S}}(\mathcal{T})$ is not.
  \end{section}


\begin{section}{Freeness, Positive Relations and the Dynamics in the Boundary}

In \cite{DaRo16}, the study of the algorithmic problem of checking whether an automaton group is free has been initiated. Formally, let \DecProblem{Freeness} be the decision problem:
     \problem
        {a $G$-automaton $\mathcal{T}$}
        {is $\mathscr{G}(\mathcal{T})$ free?}
\noindent{}One of the results from \cite{DaRo16} on \DecProblem{Freeness} is a connection to the existence of finite Schreier graphs, which we recall in the next corollary.
\begin{cor}\cite[Corollary 2]{DaRo16}
The algorithmic problem of establishing whether a group generated by a $G$-automaton $\mathcal{T}=(Q, \Sigma, \delta)$ is not free is equivalent to the problem of checking whether $\partial(\mathcal{T} \sqcup \inverse{\mathcal{T}})$ possesses a finite orbital graph in the boundary centered at a periodic point $y^{\omega} \in (Q \sqcup \inverse{Q})^\omega$ where $y$ is different to the identity in the free group $FG(Q)$. Furthermore, for a bi-reversible $G$-automaton $\mathcal{T} = (Q, \Sigma, \delta)$, checking if $\mathscr{G}(\mathcal{T})$ is free is equivalent to checking whether or not there exists a finite Schreier graph of $\partial(\mathcal{T} \sqcup \inverse{\mathcal{T}})$ centered at an essentially non-trivial element\footnote{For an $\omega$-word $\xi \in (Q \sqcup \inverse{Q})^\omega$, consider the sequence of the prefixes of $\xi$ reduced in the free group $FG(Q)$. If this sequence contains a sub-sequence converging to an element from $(Q \sqcup \inverse{Q})^\omega$ (and not to an element from $(Q \sqcup \inverse{Q})^*$), then $\xi$ is \emph{essentially non-trivial}. See \cite[Proposition 4]{DaRo16} for equivalent definitions of essentially trivial elements.} from $(Q \sqcup \inverse{Q})^\omega$.
\end{cor}

\paragraph{Positive Relations.}
In the direction of \DecProblem{Freeness}, we may consider the weaker problem \DecProblem{Positive Relations}. For a given $G$-automaton $\mathcal{T} = (Q, \Sigma, \delta)$, the set of \emph{positive relations} of $\mathscr{G}(\mathcal{T})$ with respect to $\mathcal{T}$ is given by
\[
  \mathcal{P}(\mathcal{T}) = \{ \bm{q} \in Q^{+} \mid \bm{q} = \idGrp \text{ in } \mathscr{G}(\mathcal{T}) \} \text{,}
\]
i.\,e.\ it is the set of state sequences (which do not contain states from the inverse of $\mathcal{T}$) which act like the identity. So, for instance, if $\mathcal{P}(\mathcal{T}) = \emptyset$, then the associated semigroup $\mathscr{S}(\mathcal{T})$ is torsion-free. This allows us to formally define the decision problem \DecProblem{Positive Relations}:
\problem
  {a $G$-automaton $\mathcal{T}$}
  {is $\mathcal{P}(\mathcal{T}) = \emptyset$?}

Notice that, in contrast to freeness, which is a property of the group, positive relations are dependent on the automaton representation. For example, if $\mathcal{T}$ denotes the adding machine from \autoref{expl:automatonForZ}, then $\mathscr{G}(\mathcal{T})$ is the free group of rank one, which is obviously free, but the set $\mathcal{P}(\mathcal{T})$ of positive relations is not empty, since we have, for example, $Q \ni +0 = \idGrp$ in $\mathscr{G}(\mathcal{T})$. So, the presence of a sink state in a $G$-automaton will always cause the set of positive relations to be non-empty.

Despite its dependency on the presentation, studying positive relations is still worthwhile as, for example, the emptiness of the set $\mathcal{P}(\mathcal{T})$ is strictly related to the dynamics of an automaton group $\mathscr{G}(\mathcal{T})$. While we refer the reader to \cite{DaGoKPRo16} for further details, especially regarding reversible and bi-reversible automata, we give an example of this connection.

In \cite[Lemma 5.8]{DaGoKPRo16}, it is observed that, for a $G$-automaton $\mathcal{T} = (Q, \Sigma, \delta)$, the absence of positive relations $\mathcal{P}(\mathcal{T}) = \emptyset$ implies the existence of a periodic point $u^{\omega}$, for some $u \in \Sigma^{*}$, with a non-trivial \emph{semigroup} stabilizer $\Stab_{\mathcal{T}}(u^{\omega}) \circ{}\! \neq \emptyset$.

Another example of the connection is \cite[Corollary 4]{DaRo16}: if a reversible $G$-automaton $\mathcal{T}$ generates an infinite group $G$, then the index $[G : \Stab_{\mathcal{T} \sqcup \inverse{\mathcal{T}}}(u^\omega) \circ{}\!]$ is infinite for all $u \in \Sigma^*$ if and only if the dual of $\mathcal{T}$ admits no positive relations: $\mathcal{P}(\partial\mathcal{T}) = \emptyset$.

\paragraph*{Positive Relations and the Structure of Orbital Graphs.}
Many interesting properties of automaton structures are reflected in the structure of their orbital graphs. For example, \autoref{prop:uniformly bounded} in \autoref{sct:finiteness} relates the finiteness of an automaton semigroup to the uniform boundedness of its orbital graphs. Similarly, the absence of positive relations in an automaton group is related to the absence of circles in some of its orbital graphs. We will develop this simple, yet interesting, connection next.

For any countable group $G$, we can identify each subgroup $H$ with its characteristic function, which maps an element $g \in G$ to $1$ if $g \in H$ and to $0$, otherwise. In this way, we have identified the set $\Sub(G)$ of subgroups of $G$ with $\{ 0, 1 \}^{|G|}$, which we can in turn endow with the Tychonoff topology of the $|G|$-fold product of the discrete topological space $\{ 0, 1 \}$.  For a $G$-automaton $\mathcal{T} = (Q, \Sigma, \delta)$, this allows us to consider the continuity of the map $\Stab_{\mathcal{T} \sqcup \inverse{\mathcal{T}}}(\cdot) \circ{}\!: \Sigma^{\omega} \to \Sub(\mathscr{G}(\mathcal{T}))$ which maps $\xi$ to $\Stab_{\mathcal{T} \sqcup \inverse{\mathcal{T}}}(\xi) \circ{}\!$.\footnote{For the special case of Grigorchuk's group, the map was studied by Vorobets \cite{vorobets}. Here, we consider the more general case of arbitrary automaton groups.} We say an $\omega$-word $\xi \in \Sigma^\omega$ is \emph{singular} if $\Stab_{\mathcal{T} \sqcup \inverse{\mathcal{T}}}(\cdot) \circ{}\!$ is not continuous at $\xi$. In \cite[Theorem 4.4]{DaGoKPRo16} it is shown that, for every $G$-automaton, the set~$\kappa$ of singular points has measure zero. Thus, if we denote by $\Theta=\Sigma^{\omega}\setminus\kappa$ the set of non-singular points, then $\Theta$ has measure one. Non-singular points can be characterized in the following way (see \cite[Lemma 4.2]{DaGoKPRo16}).
\begin{lemma}\label{lem:nonSingularPoints}
  Let $\mathcal{T} = (Q, \Sigma, \delta)$ be a $G$-automaton. An element $\xi\in\Sigma^{\omega}$ is not singular if and only if, for all $\wt{\bm{q}} \in Stab_{\mathcal{T} \sqcup \inverse{\mathcal{T}}}(\xi)$, there exists a prefix $u \in \Sigma^*$ of $\xi$ such that $\wt{\bm{q}} \cdot u \circ{}\! = \idGrp$.
\end{lemma}

Using this characterization, we can prove the following result on the set of non-singular points.
\begin{lemma}\label{lem:nonSingularsAreGStable}
  For every $G$-automaton $\mathcal{T} = (Q, \Sigma, \delta)$, the set $\Theta$ of non-singular points is $\mathscr{G}(\mathcal{T})$-stable, i.\,e.\ we have $\wt{\bm{q}} \circ \xi \in \Theta$ for all $\xi \in \Theta$ and all $\wt{\bm{q}} \in (Q \sqcup \inverse{Q})^*$, where $\inverse{Q}$ is the state set of $\inverse{\mathcal{T}}$.
\end{lemma}
\begin{proof}
  Let $\xi \in \Theta$ and $\wt{\bm{q}} \in (Q \sqcup \inverse{Q})^*$ and define $\eta = \wt{\bm{q}} \circ \xi$. We need to show $\eta \in \Theta$. For any $\wt{\bm{p}} \in \Stab_{\mathcal{T} \sqcup \inverse{\mathcal{T}}}(\eta)$, we have $\wt{\bm{q}}^{-1} \wt{\bm{p}} \wt{\bm{q}} \in \Stab_{\mathcal{T} \sqcup \inverse{\mathcal{T}}}(\xi)$. Thus, by the characterization in \autoref{lem:nonSingularPoints}, there is a prefix $u$ of $\xi$ such that $\wt{\bm{q}}^{-1} \wt{\bm{p}} \wt{\bm{q}} \cdot u \circ{}\! = \idGrp$. Hence, $v = \wt{\bm{q}} \circ u$ is a prefix of $\eta$ and we have $\wt{\bm{p}} \in \Stab_{\mathcal{T} \sqcup \inverse{\mathcal{T}}}(v)$. Now, for $\wt{\bm{r}} = \wt{\bm{q}} \cdot u$, we have
  \begin{align*}
    \wt{\bm{q}}^{-1} \wt{\bm{p}} \wt{\bm{q}} \cdot u
    &= [\wt{\bm{q}}^{-1} \wt{\bm{p}} \cdot (\wt{\bm{q}} \circ u)] [\wt{\bm{q}} \cdot u]
    = [\wt{\bm{q}}^{-1} \wt{\bm{p}} \cdot v] [\wt{\bm{q}} \cdot u]
    = [\wt{\bm{q}}^{-1} \cdot (\wt{\bm{p}} \circ v)] [\wt{\bm{p}} \cdot v] [\wt{\bm{q}} \cdot u]\\
    &= [\wt{\bm{q}}^{-1} \cdot v] [\wt{\bm{p}} \cdot v] [\wt{\bm{q}} \cdot u] = \bm{r}^{-1} [\wt{\bm{p}} \cdot v] \bm{r} \text{;}
  \end{align*}
  for the last step, notice the equality $(\wt{\bm{q}} \cdot u)^{-1} = \wt{\bm{q}}^{-1} \cdot v$. From $\idGrp = \wt{\bm{q}}^{-1} \wt{\bm{p}} \wt{\bm{q}} \cdot u \circ{}\!$, thus, also follows $\wt{\bm{p}} \cdot v \circ{}\! = \idGrp$. Since we have chosen $\wt{\bm{p}}$ arbitrarily from $\Stab_{\mathcal{T} \sqcup \inverse{\mathcal{T}}}(\eta)$, we have $\eta \in \Theta$ by \autoref{lem:nonSingularPoints}.
\end{proof}

Now, we are prepared to characterize the absence of positive relations in terms of orbital graphs.
\begin{prop}\label{prop: characterization empty}
  For every $G$-automaton $\mathcal{T} = (Q, \Sigma, \delta)$, the following are equivalent:
  \begin{enumerate}[label=(\roman*)]
    \item\label{itm:orbitalCycle} There is a non-singular element $\xi \in \Sigma^{\omega}$ such that $\mathcal{T} \circ \xi$ contains a (non-trivial) cycle;
    \item\label{itm:positiveRelation}$\mathcal{P}(\mathcal{T}) \neq \emptyset$;
  \end{enumerate}
\end{prop}
\begin{proof}
  \ref{itm:positiveRelation}$\implies$\ref{itm:orbitalCycle}. If $\bm{q} \in \mathcal{P}(\mathcal{T}) \neq \emptyset$, then $\mathcal{T} \circ \xi$ contains the circle $\transa{\xi}{\bm{q}}{\xi}$ for every $\xi \in \Sigma^{\omega}$. So, in particular, it contains a cycle for every non-singular $\xi$.

  \ref{itm:orbitalCycle}$\implies$\ref{itm:positiveRelation}. Suppose that $\mathcal{T} \circ \eta$ for some non-singular $\eta \in \Sigma^{\omega}$ contains a cycle $\transa{\xi}{\bm{p}}{\xi}$ at a node $\xi \in Q^* \circ \eta$ labeled by $\bm{p} \in Q^+$. We have $\xi = \bm{q} \circ \eta$ for some $\bm{q} \in Q^*$ and, thus, by \autoref{lem:nonSingularsAreGStable}, $\xi$ is also non-singular. Therefore, by \autoref{lem:nonSingularPoints}, there is prefix $u \in \Sigma^*$ of $\xi$ such that $\bm{q} \cdot u = \idGrp$, i.\,e.\ we have $\bm{q} \cdot u \in \mathcal{P}(\mathcal{T}) \neq \emptyset$.
\end{proof}

\paragraph{Undecidability of \DecProblem{Positive Relations}}
In \cite{DaGoKPRo16}, it was left open whether or not \DecProblem{Positive Relations} is undecidable. In this section, we will prove undecidability of \DecProblem{Positive Relations} and the freeness problem for automaton semigroups. Crucial to our proof is a construction by \u Suni\'c and Ventura \cite{Su-Ve09}\footnote{In fact, it is a rediscovery of a result by Brunner and Sidki \cite{Br-Sid}.}. Any pair of a $d \times d$ matrix $M$ and a vector $\bm{v} \in \mathbb{Z}^d$ gives rise to an \emph{affine transformation}
\begin{align*}
  M_{\bm{v}}: \mathbb{Z}^d &\to \mathbb{Z}^d\\
  \bm{u} &\mapsto \bm{v} + M \bm{u} \text{.}
\end{align*}
We denote by $\SAff_d(\mathbb{Z})$ the semigroup of all these affine transformations. Since the affine transformation $M_{\bm{v}}$ is invertible if and only if the matrix $M$ is, we can also define $\Aff_d(\mathbb{Z})$, the group of affine transformations of $\mathbb{Z}^d$.

Although \u Suni\'c and Ventura only considered invertible matrices, their general construction also yields the following lemma, which considers non-invertible matrices.
\begin{lemma}\label{lem:SunicVenturaAutomaton}
  Let $\mathcal{M}$ be a finite set of $\mathbb{Z}^{d \times d}$ matrices. Furthermore, for every $M \in \mathcal{M}$, let $V_M$ be the finite set of vectors $\bm{v} \in \mathbb{Z}^d$ such that all components of $\bm{v}$ are between $- \| M \|$ and $\| M \| - 1$. Here, $\| M \|$ denotes the norm $\| M \| = \max_{1 \leq i \leq d} \sum_{j = 1}^{d} | m_{i, j}|$ where $m_{i, j}$ is the entry in the $i^\textnormal{th}$ row and $j^\textnormal{th}$ column of $M$.

  Then, one can compute a complete $S$-automaton $\mathcal{T}_{\mathcal{M}}$ with state set
  \[
    Q_{\mathcal{M}} = \{ m_{M, \bm{v}} \mid M \in \mathcal{M}, \bm{v} \in V_M \}
  \]
  such that the homomorphism $\varphi: \mathscr{S}(\mathcal{T}_{\mathcal{M}}) \to \SAff_d(\mathbb{Z})$ induced by $\varphi(m_{M, \bm{v}}) = M_{\bm{v}}$ for $M \in \mathcal{M}$ and $\bm{v} \in V_M$ is a well-defined isomorphism from $\mathscr{S}(\mathcal{T}_{\mathcal{M}})$ into the subsemigroup of $\SAff_d(\mathbb{Z})$ generated by $\{ M_{\bm{v}} \mid M \in \mathcal{M}, \bm{v} \in V_M \}$.

  If all matrices in $\mathcal{M}$ are invertible, then $\mathcal{T}_{\mathcal{M}}$ is inverse-deterministic and $\varphi$ extends to a well-defined isomorphism from $\mathscr{G}(\mathcal{T})$ into the subgroup of $\Aff_d(\mathbb{Z})$ generated by $\{ M_{\bm{v}} \mid M \in \mathcal{M}, \bm{v} \in V_M \}$.\marginline{For the group part, see the erratum below.}
\end{lemma}
\begin{proof}
  The proof is based on the construction given in \cite{Su-Ve09}, which uses $n$-adic integers. The ring $\mathbb{Z}_n$ of \emph{$n$-adic integers} is the projective limit of the rings $\mathbb{Z}/n^k \mathbb{Z}$. Its elements are right-infinite sequences $(a_k)_{k \in \mathbb{Z}_+}$ such that, for every $k$, $a_k$ is in $\mathbb{Z}/n^k \mathbb{Z}$ and $a_k \equiv a_l \bmod n^k$ for all $l \geq k$. In this representation, multiplication and addition are component-wise operations and a (normal) integer $z \in \mathbb{Z}$ is a sequence which becomes stationary with value $z$.

  Another way of representing $n$-adic integers is by their \emph{$n$-adic expansion}. Any (formal) sum $Z = \sum_{k = 0}^{\infty} d_k n^k$ where the coefficients $d_k$ are from the range $0 \leq d_k \leq n - 1$ represents an $n$-adic integer $(Z \bmod n, Z \bmod n^2, Z \bmod n^3, \dots) = (d_0, d_0 + d_1 \cdot n, d_0 + d_1 \cdot n + d_2 \cdot n^2, \dots)$. On the other hand, any $n$-adic integer $Z = (a_1, a_2, a_3, \dots)$ can be written as a sum $\sum_{k = 0}^{\infty} d_k n^k$ with $d_k = \frac{a_{k + 1} - a_k}{n^k}$ (where we set $a_0 = 0$). Notice that $a_{k + 1} - a_k \equiv 0 \bmod{n^k}$ and $a_{k + 1} < n^{k + 1}$. The $n$-adic expansion is then the $\omega$-word $d_0 d_1 d_2 \dots$.

  By using $n$-adic expansions, any vector $\bm{z} \in \mathbb{Z}_n^d$ can be represented by $d$ many $\omega$-words over the alphabet $Y_n = \{ 0, \dots, n - 1 \}$ or by one $\omega$-word over the alphabet $X_n = Y_n^d = \{ 0, \dots, n - 1 \}^d$. The latter is the encoding on which the automata constructed in \cite{Su-Ve09} act.

  By \cite[Lemma 4.5]{Su-Ve09}, there is a $G$-automaton $\mathcal{A}_{M, n}$ with state set $\{ m_{\bm{v}} \mid \bm{v} \in V_M \}$ for every invertible $\mathbb{Z}^{d \times d}$ matrix $M$ such that, for every vector $\bm{v} \in V_m \subseteq \mathbb{Z}^d \subsetneq \mathbb{Z}_n^d$, the state $m_{\bm{v}}$ acts like $M_{\bm{v}}$ extended into an affine transformation $M_{\bm{v}} : \mathbb{Z}_n^d \to \mathbb{Z}_n^d$ with $M_{\bm{v}}(\bm{w}) = \bm{v} + M \bm{w}$. Here, $n$ is a number relatively prime to the (non-zero) determinant of $M$ to preserve invertibility of $M$ over the $n$-adic integers. The described construction of the automaton is clearly computable. Furthermore, it does not depend on the invertibility of the matrix $M$; if the matrix is not invertible, the obtained automaton is not necessarily a $G$-automaton anymore, but a complete $S$-automaton whose states still act in the way described above.

  As the automaton $\mathcal{T}_{\mathcal{M}}$, we can choose the disjoint union $\bigsqcup_{M \in \mathcal{M}} \mathcal{A}_{M, n}$ for some $n$ which is coprime to all non-zero determinants of the matrices in $\mathcal{M}$.

  The semigroup part of the lemma's assertion follows if we choose $\varphi$ as the restriction of maps $\mathbb{Z}_n^d \to \mathbb{Z}_n^d$ over $n$-adic integer vectors to maps $\mathbb{Z}^d \to \mathbb{Z}^d$ of (normal) integer vectors. The only missing part, here, is to show injectivity (see also \cite[Lemma 4.1]{Su-Ve09}). For this, assume that, for some matrices $M, M' \in \mathbb{Z}^{d \times d}$ and some vectors $\bm{v}, \bm{v}' \in \mathbb{Z}^d$, $M_{\bm{v}}(\bm{u}) = M'_{\bm{v'}}(\bm{u})$ holds for all $\bm{u} \in \mathbb{Z}^d$ but, for some vector $\bm{w} \in \mathbb{Z}_n^d \setminus \mathbb{Z}^d$, we have $M_{\bm{v}}(\bm{w}) \neq M'_{\bm{v'}}(\bm{w})$. Choosing $\bm{u} = \bm{0}$ as the zero vector yields $\bm{v} = \bm{v'}$. Thus, there is some $i \in \{ 1, \dots, d \}$ such that $\sum_{j = 1}^{d} m_{i,j} w_j \neq \sum_{j = 1}^{d} m'_{i,j} w_j$ where $m_{i,1}, \dots, m_{i,d} \in \mathbb{Z}$ and $m'_{i,1}, \dots, m'_{i,d} \in \mathbb{Z}$ are the entries of the $i^{\textnormal{th}}$ row of $M$ and $M'$, respectively, and $w_1, \dots, w_d \in \mathbb{Z}_n$ are the components of $\bm{w}$. As an $n$-adic integer, we can write each $w_j$ as a sequence $(a_{j, 1}, a_{j, 2}, \dots)$ with $a_{j, k} \in \mathbb{Z}/n^k \mathbb{Z}$. As addition and multiplication in $\mathbb{Z}_n$ are component-wise in this representation, there is some $k \in \mathbb{Z}_+$ with
  \[
    \left[ \sum_{j = 1}^d m_{i,j} a_{j,k} \right] \bmod n^k \neq \left[ \sum_{j = 1}^d m'_{i,j} a_{j,k} \right] \bmod n^k \text{,}
  \]
  which implies $\sum_{j = 1}^d m_{i,j} a_{j,k} \neq \sum_{j = 1}^d m'_{i,j} a_{j,k}$. Thus, we have $M_{\bm{v}}(\bm{u}) \neq M'_{\bm{v'}}(\bm{u})$ if we choose $\bm{u} \in \mathbb{Z}^d$ in such a way that the $j^{\textnormal{th}}$ component is equal to $a_{j, k}$; this constitutes a contradiction.
\end{proof}

It turns out that the relations in the (semi)group generated by $\mathcal{T}_{\mathcal{M}}$ are closely related to those in the linear (semi)group generated by the matrices in $\mathcal{M}$. We state and prove this connection in the next lemma.

\begin{lemma}\label{lem:RelationsInSunicVenturaAutomaton}
  Let $\mathcal{M}$ be a finite set of $\mathbb{Z}^{d \times d}$ matrices and let $\mathcal{T}_{\mathcal{M}}$ be the automaton from \autoref{lem:SunicVenturaAutomaton}. Then, for any sequence $M_1 \dots M_k$ and $N_1 \dots N_l$ of matrices from $\mathcal{M}$, we have
  \[
    M_1 \dots M_k = N_1 \dots N_l \iff m_{M_1, \bm{v}_1} \dots m_{M_k, \bm{v}_k} = m_{N_1, \bm{w}_1} \dots m_{N_l, \bm{w}_l}
  \]
  for some vectors $\bm{v}_i$ and $\bm{w}_j$ such that all $m_{M_i, \bm{v}_i}$ and $m_{N_j, \bm{w}_j}$ are states in $Q_{\mathcal{M}}$.

  In particular, if we denote by $I$ the $d \times d$ identity matrix, we have
  \[
    M_1 \dots M_k = I \iff m_{M_1, \bm{v}_1} \dots m_{M_k, \bm{v}_k} \text{ acts like the identity}
  \]
  for some vectors $\bm{v}_i$ such that all $m_{M_i, \bm{v}_i}$ are states in $Q_{\mathcal{M}}$.
\end{lemma}
\begin{proof}
  The second part of the lemma follows from the first one by choosing $l = 0$. The direction from left to right of the first part follows by choosing all vectors $\bm{v}_1, \dots, \bm{v}_k$ and $\bm{w}_1, \dots, \bm{w}_l$ as the $d$-dimensional zero vector $\bm{0}$ as, then, we have
  \[
    \varphi(m_{M_1, \bm{0}} \dots m_{M_k, \bm{0}}) \bm{u} = M_1 \dots M_k \bm{u} = N_1 \dots N_l \bm{u} = \varphi(m_{N_1, \bm{0}} \dots m_{N_l, \bm{0}}) \bm{u}
  \]
  for all $\bm{u} \in \mathbb{Z}^d$ where $\varphi$ is the isomorphism from \autoref{lem:SunicVenturaAutomaton}.

  For the other direction, suppose we have $m_{M_1, \bm{v}_1} \dots m_{M_k, \bm{v}_k} = m_{N_1, \bm{w}_1} \dots m_{N_l, \bm{w}_l}$ for some vectors\footnote{Indeed, for this direction of the proof, we do not require the vectors to come from $V_{M_i}$ or $V_{N_j}$, respectively.} $\bm{v}_i$ and $\bm{w}_j$. Because the images under the isomorphism $\varphi$ must be equal as well, we have
  \[
    {\left( M_1 \right)}_{\bm{v}_1} \dots {\left( M_k \right)}_{\bm{v}_k}(\bm{u}) = {\left( N_1 \right)}_{\bm{w}_1} \dots {\left( N_l \right)}_{\bm{w}_l}(\bm{u})
  \]
  for all $\bm{u} \in \mathbb{Z}^d$. Since both sides are affine transformations, there are vectors $\bm{v}, \bm{w} \in \mathbb{Z}^d$ such that
  \begin{align*}
    \bm{v} + M_1 \dots M_k \bm{u} = {\left( M_1 \right)}_{\bm{v}_1} \dots {\left( M_k \right)}_{\bm{v}_k}(\bm{u}) = {\left( N_1 \right)}_{\bm{w}_1} \dots {\left( N_l \right)}_{\bm{w}_l}(\bm{u}) = \bm{w} + N_1 \dots N_l \bm{u}
  \end{align*}
  holds for all $\bm{u} \in \mathbb{Z}^d$. Setting $\bm{u} = \bm{0}$, we obtain $\bm{v} = \bm{w}$ and, thus, $M_1 \dots M_k \bm{u} = N_1 \dots N_l \bm{u}$ for all $\bm{u} \in \mathbb{Z}^d$.
\end{proof}

We are now prepared to prove undecidability of \DecProblem{Positive Relations}.

\begin{thm}\label{thm: positive relations}
  \DecProblem{Positive Relations} is undecidable.
\end{thm}
\begin{proof}
We reduce the \emph{Identity Correspondence Problem} (\DecProblem{ICP}) \cite{Pota-Bell} to \DecProblem{Positive Relations}. It is the following decision problem. Let $FG(\Sigma)$ denote the free group over an alphabet $\Sigma$.
     \problem
     [a binary alphabet $\Sigma=\{a,b\}$]
        {$m \in \mathbb{N}$, $\Pi=\{(s_{1}, t_{1}), (s_{2}, t_{2}),\ldots, (s_{m},t_{m})\}\subseteq FG(\Sigma)\times FG(\Sigma)$}
        {is there a finite sequence $l_{1}, l_{2}, \ldots, l_{k}$ of indices with $k \geq 1$ where $1\le l_{i}\le m$ for $i=1,\ldots, k$  such that
        \[
          s_{l_{1}}s_{l_{2}}\ldots s_{l_{k}}=t_{l_{1}}t_{l_{2}}\ldots t_{l_{k}}=\varepsilon
        \]
        holds, where $\varepsilon$ is the empty word?
        }
Since \DecProblem{ICP} is undecidable\footnote{In fact, it is already undecidable for constant $m$ with $m = 8 (n - 1)$, where $n$ is the minimal number of pairs for which the \emph{Restricted PCP} is undecidable} \cite[Theorem 11]{Pota-Bell}, this reduction proves the undecidability of \DecProblem{Positive Relations}.

We consider the usual group embedding $\rho$ of $FG(\Sigma)$ into $SL_{2}(\mathbb{Z})$ defined on $\Sigma$ by
$$
\rho(a)=\left(
\begin{array}{cc}
1 & 2\\
0 & 1
\end{array}
\right), \; \rho(b)=\left(
\begin{array}{cc}
1 & 0\\
2 & 1
\end{array}
\right),\;
 \rho(a^{-1})=\left(
\begin{array}{cc}
1 & -2\\
0 & 1
\end{array}
\right), \; \rho(b^{-1})=\left(
\begin{array}{cc}
1 & 0\\
-2 & 1
\end{array}
\right)
$$
and use the same encoding of pairs from $FG(\Sigma) \times FG(\Sigma)$ as in \cite[Theorem 13]{Pota-Bell}: for each pair $(s_{i}, t_{i})\in \Pi$, we define the block matrix
$$
H_{i}=\left(
\begin{array}{cc}
\rho(s_{i}) & O_{2}\\
O_{2} & \rho(t_{i})
\end{array}
\right)
$$
where $O_{2}$ is the $2 \times 2$ zero-matrix. Let $\mathcal{H}$ be the set of these $4 \times 4$ invertible integer matrices $\{ H_1, H_2, \dots, H_m \}$. Note that existence of a solution $l_{1}, l_{2}, \ldots, l_{k}$ for \DecProblem{ICP} is equivalent to having
$$
H_{l_{1}} H_{l_{2}}\ldots H_{l_{k}}=I
$$
where $I$ is the $4\times 4$ identity matrix.

Therefore, we have so far reduced an instance of \DecProblem{ICP} to the problem of checking whether the matrix semigroup generated by a finite set $\mathcal{H}$ of invertible integer matrices contains the identity. To reduce this problem to \DecProblem{Positive Relations}, we use the $G$-automaton $\mathcal{T}_{\mathcal{H}}$ from \autoref{lem:SunicVenturaAutomaton}. As the automaton is computable, it remains to show that $\mathcal{P}(\mathcal{T}_{\mathcal{H}})$ is non-empty if and only if the linear semigroup generated by $\mathcal{H}$ contains the identity matrix. This, however, is basically the second part of \autoref{lem:RelationsInSunicVenturaAutomaton}.
\end{proof}

From \autoref{thm: positive relations} and \autoref{prop: characterization empty} we immediately obtain the following corollary.
\begin{cor}
  Let $\mathcal{T} = (Q, \Sigma, \delta)$ be a $G$-automaton. The algorithmic problem of checking whether any of the orbital graphs $\mathcal{T} \circ \xi$ centered at a non-singular element $\xi \in \Sigma^{\omega}$ contains a cycle is undecidable.
\end{cor}

In the proof of \autoref{thm: positive relations}, we have reduced \DecProblem{ICP} to \DecProblem{Positive Relations}. However, we can use the idea of this proof also for another reduction.\marginline{The rest of this section is flawed; see the erratum below.} Re-consider (\DecProblem{Group}) \DecProblem{Freeness}, the freeness problem for automaton groups:
\problem
  {a $G$-automaton $\mathcal{T}$}
  {is $\mathscr{G}(\mathcal{T})$ free?}
\noindent{}If we make some straightforwards modifications to \autoref{lem:RelationsInSunicVenturaAutomaton} to also cover inverses (or arbitrary powers of the matrices and states in general), then we also obtain a reduction from the following variant of \DecProblem{ICP}, which we will call the \emph{Inverse Identity Correspondence Problem} (\DecProblem{IICP}) to \DecProblem{Group Freeness}.
    \problem
     [a binary alphabet $\Sigma=\{a,b\}$]
        {$m \in \mathbb{N},\Pi=\{(s_{1}, t_{1}), (s_{2}, t_{2}),\ldots, (s_{m},t_{m})\}\subseteq FG(\Sigma)\times FG(\Sigma)$}
        {is there a finite sequence $l_{1}, l_{2}, \ldots, l_{k}$ of indices with $k \geq 1$ where $1\le l_{i}\le m$ for $i=1,2,\ldots, k$, and $e_{1},e_2\ldots, e_{k}\in\{1,-1\}$ with $l_i \neq l_{i + 1}$ or $e_i = e_{i + 1}$ for $i = 1, \dots, k - 1$ such that
        \[
        s^{e_{1}}_{l_{1}}s^{e_{2}}_{l_{2}}\ldots s^{e_{k}}_{l_{k}}=t^{e_{1}}_{l_{1}}t^{e_{2}}_{l_{2}}\ldots t^{e_{k}}_{l_{k}}=\varepsilon
        \]
        holds, where $\varepsilon$ is the empty word?
        }
\noindent{}We state this result in the following proposition.
\begin{prop}
  \DecProblem{IICP} is reducible to \DecProblem{Group Freeness}. In particular, if \DecProblem{IICP} is undecidable, then so is \DecProblem{Group Freeness}.
\end{prop}
  This raises an obvious question:
  \begin{prob}
    Is \DecProblem{IICP} decidable?
  \end{prob}
  \noindent{}Another problem related to the proof of \autoref{thm: positive relations} is the following.
  \begin{prob}
    The automata devised in \cite{Br-Sid, Su-Ve09} and used in \autoref{lem:SunicVenturaAutomaton} are not reversible. Are there also (bi-)reversible $G$-automata generating these groups?
  \end{prob}
  \noindent{}This problem is interesting because a positive answer would lead to undecidability of the decision problem whether all Schreier graphs cantered at periodic $\omega$-words of a $G$-automaton are infinite:
  \problem
    {a (bi-)reversible $G$-automaton $\mathcal{T} = (Q, \Sigma, \delta)$}
    {is $\wt{Q}^{*} \circ u^{\omega}$ infinite for all finite words $u \in \Sigma^+$ (where $\wt{Q}$ is the union of the states of $\mathcal{T}$ and its inverse)?}
  The reduction is based on the fact that $\mathcal{\partial \mathcal{T}} = \emptyset$ holds if and only if the stabilizer $\Stab_{\mathcal{T} \sqcup \wt{\mathcal{T}}}(u^\omega) \circ{}\!$ has infinite index in $\mathscr{G}(\mathcal{T})$ for all $u \in \Sigma^+$ \cite[Corollary 4]{DaRo16}. So, one can reduce a strengthened version of \DecProblem{Positive Relations} where the input automaton is also (bi-)reversible to the above problem by taking the dual. This strengthened version is undecidable by the same proof as for \autoref{thm: positive relations} if one can compute (bi-)reversible $G$-automata for \autoref{lem:SunicVenturaAutomaton}.

  \paragraph*{Freeness Problem for Automaton Semigroups.}
  In\marginline{Regarding the freeness problem, see the erratum below.} the direction of the freeness problems for automaton structures, we may also consider a strengthened version of the freeness problem for automaton semigroups: \DecProblem{$G$-Semigroup Freeness}. It is the following problem.
  \problem
    {a $G$-automaton $\mathcal{T}$}
    {is $\mathscr{S}(\mathcal{T})$ free?}

  Besides its obvious connection to the freeness problem of automaton groups and semigroups, this problem is also interesting since, in literature, $G$-automata quite often generate free semigroups, see for instance \cite{DaRo13, DaRo14, DaRo16, olukoya2017growthRates, SilSte05}. In \cite[Theorem 3.2]{DaRo13}, for example, decidable sufficient conditions for reset\footnote{An $S$-automaton $\mathcal{T} = (Q, \Sigma, \delta)$ is called a \emph{reset} $S$-automaton if there is a finite word $u \in \Sigma^*$ and a state $q_0 \in Q$ with $q \cdot u = q_0$ for all $q \in Q$.} $G$-automata to generate free semigroups are presented. Moreover, it is a well known fact that groups containing free semigroups (in at least two generators) have exponential growth. Usually, it seems considerably more easy to show that an automaton generates a free semigroup than to show that it generates a free group. However, it seems to be easier to show undecidability of the freeness problem for automaton semigroups than to study the freeness problem for automaton groups as the following theorem demonstrates.
\begin{thm}\label{theo: freeness semigroup}
  \DecProblem{$G$-Semigroup Freeness} is undecidable.
\end{thm}
\begin{proof}
  Consider the \DecProblem{Matrix Semigroup Freeness} problem
  \problem
    {a finite non-empty set $\mathcal{M}$ of invertible matrices from $\mathbb{N}^{3 \times 3}$}
    {is the linear semigroup generated by $\mathcal{M}$ free?}
  \noindent{}that has been proven undecidable in \cite{KlBiSat91}. As in the proof of \autoref{thm: positive relations}, we can use the $G$-automaton $\mathcal{T}_{\mathcal{M}}$ from \autoref{lem:SunicVenturaAutomaton} to reduce \DecProblem{Matrix Semigroup Freeness} to \DecProblem{$G$-Semigroup Freeness}. By \autoref{lem:RelationsInSunicVenturaAutomaton}, the semigroup generated by $\mathcal{T}_{\mathcal{M}}$ is free if and only if so is the linear semigroup generated by $\mathcal{M}$.
\end{proof}

As \DecProblem{$G$-Semigroup Freeness} is a strengthened version of the freeness problem for automaton semigroups, we get the following corollary, which solves an open problem of Grigorchuk, Nekrashevych and Sushchansky \cite[7.2 b)]{GriNeShu}.
\begin{cor}\label{cor:SemigroupFreenessIsUndecidable}
  The freeness problem for automaton semigroups
  \problem
    {an $S$-automaton $\mathcal{T}$}
    {is $\mathscr{S}(\mathcal{T})$ free?}
  \noindent{}is undecidable.
\end{cor}

\autoref{theo: freeness semigroup} can also be interpreted outside the context of automaton structures. Remember that any $S$-automaton $\mathcal{T} = (Q, \Sigma, \delta)$ is, by definition, a synchronous, deterministic finite state transducer whose input and output alphabets coincide. In a more automata theoretic setting, one would not use these transducers to define semigroups but to define (rational) relations. For this, we need to select an initial state $q_0 \in Q$ and a set of final states $F \subseteq Q$. Under this choice, the \emph{rational relation} accepted by $\mathcal{T}$ is
\[
  \mathcal{R}[\mathcal{T}, q_0, F] = \{ (u, v) \in \Sigma^* \times \Sigma^* \mid q_0 \circ u = v, q_0 \cdot u \in F \} \text{,}
\]
i.\,e.\ a pair of finite words $(u, v)$ is in the relation if and only if, when reading $u$ starting in $q_0$, one ends in a state belonging to $F$ and the output is $v$. 

Now, \autoref{theo: freeness semigroup} yields the following undecidability result:
\begin{cor}
  The problem\footnote{Problems similar to the one presented here are also discussed in \cite{pierce}.}
  \problem
    {a synchronous, complete, deterministic, invertible, finite-state transducer $\mathcal{T}$ with coinciding input and output alphabet and state set $Q$, in which every state is final}
    {is there a $k \geq 1$ such that, in the $k^\text{th}$-power of $\mathcal{T}$, one can choose two different initial states $\bm{p}, \bm{q} \in Q^k$ with $\mathcal{R}[\mathcal{T}^k, \bm{p}, Q^k] = \mathcal{R}[\mathcal{T}^k, \bm{q}, Q^k]$?}
  \noindent{}is undecidable.
\end{cor}
\begin{proof}
  The reduction function from (the complement of) \DecProblem{$G$-Semigroup Freeness} to this problem  is the identity function with one exception: if the input automaton contains only one state, then it generates a finite and, thus, non-free semigroup; therefore, the reduction function can map these automata to an arbitrary but fixed automaton with a positive answer to the above question.

  We need to show that, for a $G$-automaton $\mathcal{T} = (Q, \Sigma, \delta)$ with at least two states, the semigroup $\mathscr{S}(\mathcal{T})$ is \textbf{not} free if and only if there is a $k \geq 1$ and $\bm{p}, \bm{q} \in Q^k$ with $\bm{p} \neq \bm{q}$ such that $\mathcal{R}[\mathcal{T}^k, \bm{p}, Q^k] = \mathcal{R}[\mathcal{T}^k, \bm{q}, Q^k]$. Notice that we have $\bm{p} \circ{}\! = \bm{q} \circ{}\!$ if and only if we have $\mathcal{R}[\mathcal{T}^{|\bm{p}|}, \bm{p}, Q^{|\bm{p}|}] = \mathcal{R}[\mathcal{T}^{|\bm{q}|}, \bm{q}, Q^{|\bm{q}|}]$ for any par $\bm{p}, \bm{q} \in Q^+$. Thus, it suffices to show that there are $\bm{p}, \bm{q} \in Q^+$ with $\bm{p} \neq \bm{q}$ but $|\bm{p}| = |\bm{q}|$ and $\bm{p} \circ{}\! = \bm{q} \circ{}\!$ if $\mathscr{S}(\mathcal{T})$ is not free.

  Suppose that $\mathscr{S}(\mathcal{T})$ is not free, i.\,e.\ that there are $\bm{p}', \bm{q}' \in Q^+$ with $\bm{p}' \neq \bm{q}'$ but $\bm{p}' \circ{}\! = \bm{q}' \circ{}\!$. We are done if we have $|\bm{p}'| = |\bm{q}'|$. Otherwise, we distinguish two cases: $\bm{p}' \bm{q}' \neq \bm{q}' \bm{p}'$ and $\bm{p}' \bm{q}' = \bm{q}' \bm{p}'$. In the former one, we can set $\bm{p} = \bm{p}' \bm{q}'$ and $\bm{q} = \bm{q}' \bm{p}'$. The latter case is a bit more complicated: as $\bm{p}'$ and $\bm{q}'$ commute, there is a word $\bm{r} \in Q^+$ such that $\bm{p}' = \bm{r}^k$ and $\bm{q}' = \bm{r}^\ell$ for some $k, \ell \in \mathbb{N}$ (see, for example, \cite[p.~8, Porposition~1.3.2]{lot83}). Without loss of generality, we may assume $k < \ell$. We have $\id = \inverse{\bm{p}'} \bm{p}' \circ{}\! = \inverse{\bm{r}}^k \bm{q}' \circ{}\! = \inverse{\bm{r}}^k \bm{r}^\ell \circ{}\! = \bm{r}^{\ell - k} \circ{}\!$, i.\,e.\ that $\bm{r}^{\ell - k} \in Q^+$ acts like the identity. Since we have $|Q| \geq 2$, there is a state $q \in Q$ which differs from the first letter of $\bm{r}$ (seen as a word over $Q$). Thus, we have $q \bm{r}^{\ell - k} \neq \bm{r}^{\ell - k} q$ but $|q \bm{r}^{\ell - k}| = |\bm{r}^{\ell - k} q|$ and $q \bm{r}^{\ell - k} \circ{}\! = \bm{r}^{\ell - k} q \circ{}\!$.
\end{proof}

\end{section}

\begin{section}{The Finiteness Problem for Invertible, Bi-Reversible Automata}\label{sct:finiteness}

In this section, we consider \DecProblem{Group Finiteness}, the finiteness problem for automaton groups:
 \problem
      {a $G$-automaton $\mathcal{T}$}
      {is $\mathscr{G}(\mathcal{T})$ finite?}
\noindent Although it has been studied widely, the decidability of this problem is still an open problem. Notice, however, that the problem is semi-decidable in the sense that there is an algorithm which stops if and only if $\mathscr{G}(\mathcal{T})$ is finite. For this algorithm, one can use the naïve approach of enumerating all state sequences (in order for ascending length) until no further new group elements are found. The result closest to proving undecidability of \DecProblem{Group Finiteness} is due to Gillibert, who showed that the finiteness problem for automaton \emph{semigroups} is undecidable \cite{Gilbert13}, and the recent result on the undecidability of the order problem for automaton groups (checking whether a group element has finite order) \cite{gillibert2017automaton}.

In this section, we extend this result to get closer to automaton groups: we show that the finiteness problem for automaton semigroups remains undecidable if the input automaton is inverse-deterministic and bi-reversible. In particular, we obtain undecidability of the problem whether a subsemigroup of an automaton-inverse semigroup given by some generating states is finite. We obtain these results by connecting the existence of an infinite orbit in the boundary to the finiteness of the semigroup.

\paragraph{Orbital Graphs and Finiteness.}
We start by characterizing finite semigroups by their dynamics on the boundary. Finiteness of an automaton semigroup is related to the sizes of its orbital graphs as we have the following result that naturally extends the group case  \cite[Corollary 1]{DaRo14}.
\begin{prop}\label{prop:uniformly bounded}
  Let $\mathcal{T} = (Q, \Sigma, \delta)$ be an $S$-automaton. Then, its generated semigroup $\mathscr{S}(\mathcal{T})$ is finite if and only if there exists a constant $C$ such that, for every $\xi \in \Sigma^\omega$, the size of $Q^{*}\circ \xi$ is bounded by $C$.
\end{prop}
\begin{proof}
  Suppose that $\mathscr{S}(\mathcal{T})$ is finite. Then, clearly, $|Q^{*}\circ \xi| \le |\mathscr{S}(\mathcal{T})|$ for all $\xi\in \Sigma^{\omega}$.

  Conversely, since $Q$ is a finite set and there is a constant $C$ such that $|Q^{*}\circ \xi| \le C$ for all $\xi\in \Sigma^{\omega}$, there are only finitely many possibilities for the orbital graph $\mathcal{T} \circ \xi$ up to edge-label preserving isomorphism of rooted graphs. For each of these isomorphism classes, we fix a representative. Let them be $\mathcal{T} \circ \xi_{1}, \ldots, \mathcal{T} \circ \xi_{n}$ and let $\mathcal{U}$ be the union of these orbital graphs. Notice that $\mathcal{U}$, as a finite union of finite graphs, is finite itself. Additionally, we have a partial action of $Q^+$ on $\mathcal{U}$: For any $\bm{q} \in Q^+$ with $\bm{q} = q_\ell \dots q_1$ for $q_1, \dots, q_\ell \in Q$, let $\bm{q} \star u$ denote the vertex of $\mathcal{U}$ reached by following the (unique if existing) path
  \begin{center}
    \begin{tikzpicture}[auto, >=latex]
      \node (u) {$u$};
      \node[right=of u] (v1) {$q_1 \circ u$};
      \node[right=of v1] (v2) {$q_2 q_1 \circ u$};
      \node[right=of v2] (dots) {$\dots$};
      \node[right=of dots] (vl) {$q_\ell \dots q_1 \circ u = \bm{q} \star u$};
      
      \draw[->] (u) edge node {$q_1$} (v1)
                (v1) edge node {$q_2$} (v2)
                (v2) edge node {$q_3$} (dots)
                (dots) edge node {$q_\ell$} (vl);
    \end{tikzpicture}
  \end{center}
  of $\mathcal{U}$ which starts in $u$ and has label $q_1 \dots q_\ell$. The partial maps $\bm{q} \star{}\!$ mapping a vertex $u$ to $\bm{q} \star u$ form a semigroup $T$, which is finite since there are only finitely many partial maps from the finite vertex set of $\mathcal{U}$ to itself.

  We show that $\mathscr{S}(\mathcal{T})$ is isomorphic to $T$ via the isomorphism $\varphi: \mathscr{S}(\mathcal{T}) \to T, \bm{q} \circ{}\! \mapsto \bm{q} \star{}\!$. First, we have to show that $\varphi$ is well-defined. For this, we need to show that $\bm{q} \circ{}\! = \bm{p} \circ{}\!$ implies $\bm{q} \star{}\! = \bm{p} \star{}\!$. Consider $\bm{q} = q_m \dots q_1$ and $\bm{p} = p_\ell \dots p_1$ for $q_1, \dots, q_m, p_1, \dots, p_\ell \in Q$ and suppose that there is some $\xi \in \Sigma^\omega$ such that $\bm{q} \star{}\!$ differs from $\bm{p} \star{}\!$ on $\xi$. There are two cases: either one of them (say: $\bm{p} \star{}\!$) is undefined on $\xi$ or they are both defined but their values differs. In the first case, there is some path in $\mathcal{U}$ which starts in $\xi$ and is labeled by $q_1 \dots q_m$; on the other hand, there is no such path labeled with $p_1 \dots p_\ell$. Since $\mathcal{U}$ arises as a union of orbital graphs the (non-)existence of these paths mean that $\bm{q} \circ{}\!$ is defined on $\xi$ while $\bm{p} \circ{}\!$ is not. For the second case, we have that both paths exist but end in different nodes. For the same reasons, this means that $\bm{q} \circ{}\!$ and $\bm{p} \circ{}\!$ are both defined on $\xi$ but differ in value.
  
  As surjectivity of $\varphi$ is trivial, it remains to show injectivity. For two distinct elements, $\bm{q} \circ{}\!$ and $\bm{p} \circ{}\!$ of $\mathscr{S}(\mathcal{T})$, there is an $\omega$-word $\xi \in \Sigma^\omega$ such that $\bm{q} \circ{}\!$ and $\bm{p} \circ{}\!$ differ on $\eta$ (either because only one of them is defined or because their respective values differ). Let $\mathcal{T} \circ \xi_i$ be the orbital graph isomorphic to $\mathcal{T} \circ \xi$. Then, we have $\bm{q} \star{}\! \neq \bm{p} \star{}\!$ because there either is (without loss of generality) only a path belonging to $\bm{q}$ which starts in $\xi_i$ in $\mathcal{U}$ but none belonging to $\bm{p}$ or both paths exist but they end in different nodes.
\end{proof}

Notice that the same idea can be used to prove a similar result for automaton-inverse semigroups:
\begin{prop}
  Let $\mathcal{T} = (Q, \Sigma, \delta)$ be an $\inverse{S}$-automaton. Then, its generated inverse semigroup $\inverse{\mathscr{S}}(\mathcal{T})$ is finite if and only if there exists a constant $C$ such that, for every $\xi \in \Sigma^\omega$, the size of $\tilde{Q}^{*}\circ \xi$ is bounded by $C$, where $\tilde{Q}$ denotes the disjoint union of the states of $\mathcal{T}$ and the states of its inverse $\inverse{\mathcal{T}}$.
\end{prop}

Notice that it does not suffice that the sizes of all $Q^* \circ \xi$ are bounded by some constant for an automaton-inverse semigroup to be finite; we also need to consider the inverses. A counter-example to this has already been given in \autoref{lem:inverseAndSemigroupOrbitsDoNotCoincide}.\footnote{This stops us from directly transferring our proof for the undecidability of the finiteness problem for invertible, reversible automata to the finiteness problem for automaton-inverse semigroups.}

Given an infinite automaton semigroup, \autoref{prop:uniformly bounded} states that, for every $n \in \mathbb{N}$, there is a word whose orbit is larger than $n$. While it seems plausible that one can obtain a single $\omega$-word with an infinite orbit from this increasing sequence, it is not obvious how this can be done. For example, consider the $S$-automaton
\begin{center}
  \begin{tikzpicture}[auto, shorten >=1pt, >=latex]
    \node[state] (q) {$q$};
    \node[state, right=of q] (p) {$p$};
    
    \path[->] (q) edge[loop left] node {$b/b$} (q)
                  edge node {$a/b$} (p)
              (p) edge[loop right] node[align=center] {$a/a$\\$b/b$} (p)
    ;
  \end{tikzpicture}.
\end{center}
While $p$ acts like the identity, the action of $q$ is to replace the first $a$ in the input word with $b$. Thus, the orbit of $a^n$ contains $n + 1$ words. However, the same is true for $b^n a^n$. Now, while the first sequence $(a^n)_{n \in \mathbb{N}}$ converges (with respect to the prefix metric) to $a^\omega$, the second sequence $(b^n a^n)_{n \in \mathbb{N}}$ converges to $b^\omega$. The former indeed has an infinite orbit. The orbit of the latter, however, only contains the word $b^\omega$ itself. This example demonstrates that we cannot simply take the limit point of an arbitrary sequence of words with increasing orbit size and obtain a(n $\omega$-)word with an infinite orbit, which leads us to the following question.
\begin{prob}\footnote{Shortly after making the first version of this paper available on the arXiv (arXiv:1712.07408v1), this problem was solved independently by the authors and Dominik Francoeur \cite{francoeur2018infinite}, which resulted in a joint paper on this topic merging the two proofs \cite{dangeli2019orbits}.}\enlargethispage{2\baselineskip}
  Let $\mathcal{T} = (Q, \Sigma, \delta)$ be an $S$-automaton. Is $\mathscr{S}(\mathcal{T})$ infinite if and only if there is an $\omega$-word $\xi \in \Sigma^\omega$ with an infinite orbit $Q^* \circ \xi$?

  If $\mathcal{T}$ is a $G$-automaton, is $\mathscr{G}(\mathcal{T})$ then infinite if and only if there is an $\omega$-word $\xi \in \Sigma^\omega$ with an infinite orbit $\tilde{Q}^* \circ \xi$ (where $\tilde{Q}$ is the disjoint union of the states of $\mathcal{T}$ and the states of its inverse)?
\end{prob}
\noindent{}A positive answer to the semigroup case leads to a positive answer for the group case by \autoref{lem:groupAndSemigroupOrbitsCoincide} since $\mathscr{S}(\mathcal{T})$ is finite if and only if so is $\mathscr{G}(\mathcal{T})$ \cite{aklmp12}.

\paragraph{Wang Tilings and Automata.}

Recall that a Wang tile is a unit square tile with a color on each edge. Formally, it is a quadruple $t = (t_{W}, t_{S}, t_{E}, t_{N})\in C^{4}$ where $C$ is a finite set of colors. The choice of $W$, $S$, $E$ and $N$ stems from the four cardinal directions and we say, for example, that $t_W$ is the color at the \emph{west side} of $t$. When seeing it as a Wang tile, we write the tuple $t = (t_{W}, t_{S}, t_{E}, t_{N})$ as $t = \wang{t_W}{t_S}{t_E}{t_N}$.

A tile set is a finite set~$\mathcal{W}$ of Wang tiles. For each $t\in \mathcal{W}$ and~$D \in \{ W, S, E, N \}$, we let~$t_{D}$ denote the color of the edge on the $D$-side. A tile set $\mathcal{W}$ \emph{tiles the discrete plane} $\mathbb{Z}^2$ if there is a map $f: \mathbb{Z}^{2} \to \mathcal{W}$ that associates to each point in the discrete plane a tile from~$\mathcal{W}$ such that adjacent tiles share the same color on their common edge, i.\,e.\ $f(x,y)_{E}=f(x+1,y)_{W}$ and $f(x,y)_{N}=f(x,y+1)_{S}$ for every $(x,y)\in\mathbb{Z}^{2}$. Such a map is called a $\mathbb{Z}^2$-tiling. Analogously, a tile set $\mathcal{W}$ \emph{tiles the first quadrant of the discrete plane} $\mathbb{N}^2$ if it admits an $\mathbb{N}^2$-tiling, i.\,e.\ there is a map $g: \mathbb{N}^2 \to \mathcal{W}$ with $g(x,y)_{E}=g(x+1,y)_{W}$ and $g(x,y)_{N}=g(x,y+1)_{S}$ for every $(x,y)\in\mathbb{N}^{2}$.

A $\mathbb{Z}^2$-tiling $f$ is called \emph{periodic} if there exists a (non-zero) periodicity vector $\bm{v} \in \mathbb{Z}^2$ for $f$, i.\,e.\ we have $f(\bm{t} + \bm{v}) = f(\bm{t})$ for all $t \in \mathbb{Z}^2$. If a tile set does not admit a periodic $\mathbb{Z}^2$-tiling, it is called \emph{aperiodic}.

For an $\mathbb{N}^2$-tiling $g$, we may define the notion of horizontal words: the \emph{$i^\text{th}$ horizontal word of $g$} is $f(0, i)_S f(1, i)_S \dots$, i.\,e.\ the $\omega$-word over $C$ given by the south colors of the rectangle $[0, \infty] \times [i, i]$. We call $g$ \emph{y-recurrent} if there are $i \neq j$ such that the $i^\text{th}$ and $j^\text{th}$ horizontal word of $g$ coincide; if no such $i$ and $j$ exits, we call $g$ \emph{non-y-recurrent}. The notions of y-recurrence and periodicity of a tile set are linked:\footnote{The proof is basically the same as the one showing that $\mathcal{W}$ tiles $\mathbb{Z}^2$ if and only if it tiles $\mathbb{N}^2$ (see \cite{robinson1971undecidability}).}

\begin{lemma}\label{lem: recurrent implies periodic}
  A tile set $\mathcal{W}$ admits a y-recurrent $\mathbb{N}^2$-tiling if and only if it admits a periodic $\mathbb{Z}^2$-tiling.
\end{lemma}
\begin{proof}
  Let $g$ be a y-recurrent $\mathbb{N}^2$-tiling for $\mathcal{W}$. Then, by definition, there are $i < j$ such that the $i^\text{th}$ horizontal word is the same as the $j^\text{th}$ one, i.\,e.\ the rectangle $[0, \infty] \times [i, j - 1]$ is colored the same way on its north and on its south side. Notice that there are only finitely many possible colorings for the west sides of the rectangles $[k, \infty] \times [i, j - 1]$ for $k \in \mathbb{N}$. Therefore, there must be $k, \ell \in \mathbb{N}$ with $k < \ell$ such that the colorings on the west side of the respective rectangles coincide, or, in other words, such that the rectangle $[k, \ell - 1] \times [i, j - 1]$ has the same coloring on its west and its east side. Thus, we can repeat the tile pattern associated to this rectangle infinitely often, both, horizontally and vertically, which yields a periodic $\mathbb{Z}^2$-tiling.
  
  For the other direction, let $f$ be a periodic $\mathbb{Z}^2$-tiling for $\mathcal{W}$. By an argument similar to the one just presented, one can see that $f$ can be assumed to be vertically periodic, i.\,e.\ that there is a $v_y > 0$ such that $f(x, y) = f(x, y + v_y)$ for all $x, y \in \mathbb{Z}$ (see, for example, \cite{robinson1971undecidability}). Clearly, the restriction of $f$ into a map $\mathbb{N}^2 \to \mathcal{W}$ is a y-recurrent $\mathbb{N}^2$-tiling.
\end{proof}
\noindent{}Notice that the periodic $\mathbb{Z}^2$-tiling might be different from the y-recurrent $\mathbb{N}^2$-tiling.

There is a natural way to associate a (possibly non-deterministic) automaton $\mathcal{T}(\mathcal{W})$ to a tile set $\mathcal{W}$ (and vice-verse, see \cite{DaGoKPRo16}): the state set and alphabet are the set of colors and we associate a transition $\trans{t_W}{t_S}{t_N}{t_E}$ to every tile $\mathcal{W} \ni t = \wang{t_W}{t_S}{t_E}{t_N}$. We say that a tile set~$\mathcal{W}$ is $CD$-deterministic with~$(C,D) \in \{ (S, W), (S, E), (N, W), (N, E)\}$ if each tile~$t\in\mathcal{W}$ is uniquely determined by its pair~$(t_{C}, t_{D})$ of colors on the $C$ and $D$ sides. Whenever~$\mathcal{W}$ is $CD$-deterministic for each~$(C,D) \in \{ (S, W), (S, E), (N, W), (N, E)\}$, we say that $\mathcal{W}$ is \emph{4-way deterministic}. There are some obvious connections between the determinism of a tile set and its associated automaton, which we list in the following lemma (see also \cite[Lemma 6.1]{DaGoKPRo16}).
\begin{lemma}\label{lem: basic facts tilings automata}
  The following facts hold:
  \begin{itemize}
    \item $\mathcal{W}$ is $SW$-deterministic if and only if $\mathcal{T}(\mathcal{W})$ is deterministic;
    \item $\mathcal{W}$ is $SE$-deterministic if and only if $\mathcal{T}(\mathcal{W})$ is reversible;
    \item $\mathcal{W}$ is $NW$-deterministic if and only if $\mathcal{T}(\mathcal{W})$ is inverse-deterministic;
    \item $\mathcal{W}$ is $NE$-deterministic if and only if $\mathcal{T}(\mathcal{W})$ is inverse-reversible.
  \end{itemize}
\end{lemma}

For aperiodic tile sets, the notion of y-recurrence and the finiteness of the associated semigroup are linked as the following proposition demonstrates.
\begin{prop}\label{prop:recurrentTiling}
  Let $\mathcal{W}$ be a $SW$-deterministic, aperiodic tile set and let $\mathcal{T} = \mathcal{T}(W) = (Q, \Sigma, \delta)$ be the associated $S$-automaton. Then,
  \begin{itemize}
    \item $\mathscr{S}(\mathcal{T})$ is infinite
    \item $\mathcal{W}$ admits a $\mathbb{Z}^2$-tiling and
    \item there is an $\omega$-word $\xi \in \Sigma^\omega$ with an infinite orbit $Q^* \circ \xi$
  \end{itemize}
  are equivalent.
\end{prop}
\begin{proof}
  First, we show that a $\mathbb{Z}^2$-tiling induces an $\omega$-word with an infinite orbit, which means that $\mathscr{S}(\mathcal{T})$ must be infinite (by \autoref{prop:uniformly bounded}). Therefore, suppose that $\mathcal{W}$ admits a $\mathbb{Z}^2$-tiling $f$. Since $\mathcal{W}$ is aperiodic, its restriction $g$ into an $\mathbb{N}^2$-tiling must be non-y-recurrent by \autoref{lem: recurrent implies periodic}. Let $\xi_i$ denote the $i^\text{th}$ horizontal word of $g$ and let $q_i = g(0, i)_W$ for all $i \in \mathbb{N}$ (i.\,e.\ $q_1 q_2 \dots$ labels the west side of the first quadrant read from bottom to top). By construction of $\mathcal{T}$, we have $\xi_{i + 1} = q_i \circ \xi_i$ and, since $g$ is non-y-recurrent, we have $\xi_i \neq \xi_j$ for all $i \neq j$. Therefore, the orbit $Q^* \circ \xi_0$ is infinite.

  Next, assume that $\mathscr{S}(\mathcal{T})$ is infinite. We show that $\mathcal{W}$ admits a $\mathbb{Z}^2$-tiling tiling in this case, which -- as shown above -- implies that there is an $\omega$-word with an infinite orbit. In fact, we only need to show that $\mathcal{W}$ admits tilings for arbitrarily large squares $[0, \ell] \times [0, \ell]$ since, then, it also admits a $\mathbb{Z}^2$-tiling (by a standard compactness argument, see e.\,g.\ \cite{robinson1971undecidability}). By \autoref{prop:uniformly bounded} there is an infinite sequence $(\xi_{i})_{i \in \mathbb{N}}$ of $\omega$-words such that $| Q^* \circ \xi_i | \geq i$, which, by definition, is also the size of the orbital graph $\mathcal{T} \circ \xi_i$. Notice that, for every $i \in \mathbb{N}$, all nodes in $\mathcal{T} \circ \xi_i$ have out-degree at most $|Q|$. Thus, for every $\ell \in \mathbb{N}$, we find an $i(\ell) \in \mathbb{N}$ such that $\mathcal{T} \circ \xi_{i(\ell)}$ contains a path of length $\ell$
  \[
    \eta^{(\ell)} = \eta^{(\ell)}_0 \overset{q^{(\ell)}_1}{\longrightarrow} \eta^{(\ell)}_1 \overset{q^{(\ell)}_2}{\longrightarrow} \dots \overset{q^{(\ell)}_{\ell}}{\longrightarrow} \eta^{(\ell)}_{\ell} \text{,}
  \]
  where all $\eta^{(\ell)}_j \in \Sigma^\omega$ are in the orbit $Q^* \circ \xi_{i(\ell)}$. Let $u_j^{(\ell)}$ denote the prefix of $\eta^{(\ell)}_j$ of length $\ell$. Notice that, due to prefix-compatibility, we still have $u^{(\ell)}_j = q_j \circ u^{(\ell)}_{j - 1}$. By construction of $\mathcal{T} = \mathcal{T}(\mathcal{W})$, this yields a tiling of $[0, \ell] \times [0, \ell]$ where the south sides of the $i^\text{th}$ row (with $i \in \mathbb{N}$) is labeled by $u^{(\ell)}_i$, the north side of the $\ell^\text{th}$ row is labeled by $u^{(\ell)}_\ell$ and the west side of the square is labeled by $q_1 \dots q_\ell$.
\end{proof}

The previous proposition can be used for the reduction to show the main result of this section:
\begin{thm}\label{theo: undecidability finite}
  The strengthened version
  \problem
    {a bi-reversible and bi-deterministic automaton $\mathcal{T}$}
    {is $\mathscr{S}(\mathcal{T})$ finite?}
  \noindent{}of the finiteness problem for automaton semigroups remains undecidable. In particular, it is undecidable to check whether an automaton subsemigroup of an automaton-inverse semigroup is finite.
  
  Furthermore, the problem
  \problem
    {a bi-reversible and bi-deterministic automaton $\mathcal{T} = (Q, \Sigma, \delta)$}
    {is there an $\omega$-word $\xi \in \Sigma^\omega$ with an infinite orbit $Q^* \circ \xi$?}
  \noindent{}is undecidable.
\end{thm}
\begin{proof}
  Lukkarila \cite{Lukkarila} has shown that the tiling problem for Wang tilings remains undecidable when the input is restricted to be a $4$-way deterministic, aperiodic tile sets. In other words, he showed undecidability of the problem:
  \problem
    {a $4$-way deterministic, aperiodic tile set $\mathcal{W}$}
    {does $\mathcal{W}$ admit a $\mathbb{Z}^2$-tiling?}
  To prove this undecidability, he constructed a $4$-way deterministic tile set from a Turing Machine in such a way that the tile set admits a tiling if and only if the Turing Machine halts. The constructed tiles consist of multiple layers of tiles (using direct products) and the first layer uses the 4-way deterministic, aperiodic tile set of Kari and Papasoglu \cite{Ka-Pa}.
  
  For a 4-way-deterministic and aperiodic tile set $\mathcal{W}$, the automaton $\mathcal{T}(\mathcal{W})$ is bi-reversible and bi-deterministic by \autoref{lem: basic facts tilings automata} and, by \autoref{prop:recurrentTiling}, it generates an infinite semigroup if and only if $\mathcal{W}$ admits a $\mathbb{Z}^2$-tiling, which is the case if and only if there is an $\omega$-word with an infinite orbit. Thus, mapping $\mathcal{W}$ to $\mathcal{T}(\mathcal{W})$ is a co-reduction from Lukkarila's problem to our strengthened version of the finiteness problem and a reduction to our second problem.
\end{proof}

Unfortunately, the above theorem neither shows undecidability for the finiteness problem of automaton groups nor for the finiteness problem of automaton-inverse semigroups. Although, for the former, we have that $\mathscr{G}(\mathcal{T})$ is infinite if and only if $\mathscr{S}(\mathcal{T})$ is infinite, the automaton associated to a 4-way deterministic tile set is not necessarily complete and, thus, not a $G$-automaton. This incompleteness is inherent to the undecidability of the tiling problem: the automaton is complete if and only if, for every pair $(c_S, c_W)$ of colors, there is at least (and, thus, due to 4-way determinism exactly) one tile whose south side is colored with $c_S$ and whose west side is colored with $c_W$. If this is true for some tile set, however, it always admits a $\mathbb{N}^2$-tiling as one can choose the colors for the x-axis and for the y-axis arbitrarily and continue the tiling from there.

One can complete an $S$-automaton by adding a sink state (see \cite[Proposition 1]{DAngeli2017}). This has the effect of (possibly\footnote{Notice that -- although it is wrongfully stated in the proof of \cite[Proposition 1]{DAngeli2017} -- this is not always the case; see \cite[Section~3]{structurePart} for a discussion.}) adding a zero to the generated semigroup and, therefore, maintains (in)finiteness. However, it does not maintain reversibility, which leads to the following open problem.
\begin{prob}
  Given a reversible $S$-automaton, can one compute a complete and reversible $S$-automaton $\mathcal{T}'$ such that $\mathscr{S}(\mathcal{T})$ is finite if and only if $\mathscr{S}(\mathcal{T}')$ is finite? Is it possible to give such a construction if $\mathcal{T}$ is a bi-reversible $\inverse{S}$-automaton?
\end{prob}
\noindent{}A positive answer to this open problem would lead to the undecidability of the finiteness problem for automaton groups. Indeed, for the reduction, one could take the bi-reversible and bi-deterministic automaton $\mathcal{T}$ obtained from the 4-way-deterministic tile set in \autoref{theo: undecidability finite}, compute its reversible completion $\mathcal{T}'$ and take the dual of $\mathcal{T}'$, which is a $G$-automaton. Then, we had
\[
  |\mathscr{S}(\mathcal{T})| = \infty \iff |\mathscr{S}(\mathcal{T}')| = \infty \iff |\mathscr{S}(\partial\mathcal{T}')| = \infty \iff |\mathscr{G}(\partial\mathcal{T}')| = \infty \text{,}
\]
where the first equivalence is due to the assumption in the open problem. For the second equivalence, see e.\,g.\ \cite[Corollary 1]{DaRo14} or \cite{aklmp12} and, for the third equivalence see \cite{aklmp12}.

For automaton-inverse semigroups, on the other hand, the partiality is not problematic. Here, the problem is rather that we have to consider the inverses as well. From the perspective of Wang tiles, taking the inverse belongs to mirroring a tile at the horizontal axis. The union of the tiles and their mirrored versions is not an aperiodic tile set anymore. In fact, it will always admit an $\mathbb{N}^2$-tiling as long as we can tile a single right-infinite row. In this case, our approach of obtaining an infinite orbit cannot be applied, which leaves us with the following open problem.

\begin{prob}
  Is the finiteness problem for automaton-inverse semigroups
  \problem
    {an $\inverse{S}$-automaton}
    {is $\inverse{\mathscr{S}}(\mathcal{T})$ finite?}
  \noindent{}decidable?
\end{prob}

\end{section}

\bibliographystyle{plain}
\bibliography{references}

\cleardoublepage
\includepdf[pages=-]{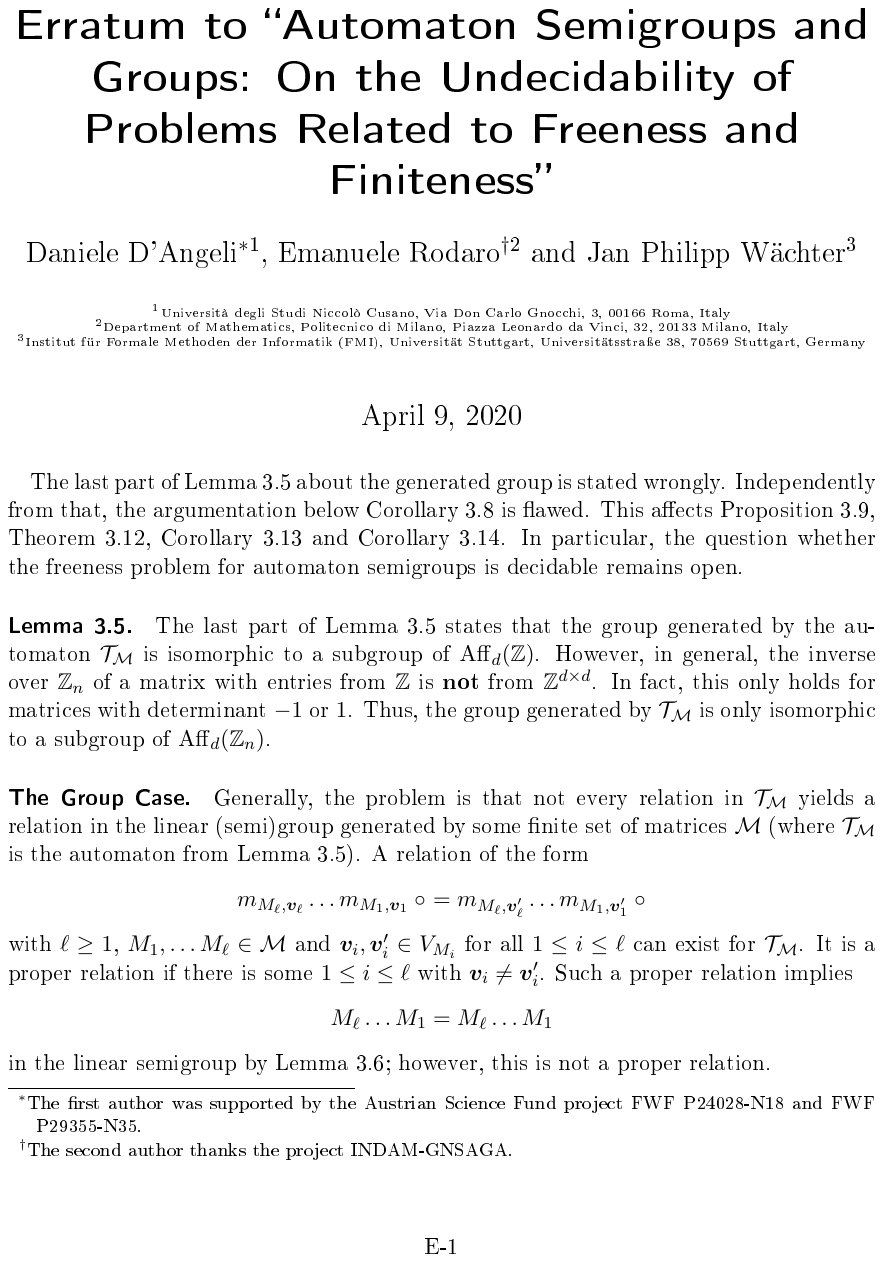}

\end{document}